\documentclass[runningheads,a4paper]{llncs}
\usepackage{fullpage}
\usepackage{algorithm}
\usepackage{algorithmic}
\usepackage{amssymb}
\usepackage{graphicx}
\usepackage{tikz} \usetikzlibrary{chains,positioning,scopes}
\newcommand{\keywords}[1]{\par\addvspace\baselineskip
\noindent\keywordname\enspace\ignorespaces#1}

\begin{document}

\mainmatter  
\title{Social Influence as a Voting System: a Complexity Analysis of Parameters and Properties\thanks{This work is partially supported by grant 2009SGR1137 (ALBCOM) of ``Generalitat de Catalunya''. X.~Molinero is partially funded by grant MTM2012--34426 of the ``Spanish Economy and Competitiveness Ministry''.  F.~Riquelme is supported by grant BecasChile of the ``National Commission for Scientific and Technological Research of Chile'' (CONICYT). M.~Serna is supported by grant TIN2007--66523 (FORMALISM)  of  the ``Ministerio de Ciencia e Inovaci\'on y el Fondo Europeo de Desarrollo Regional''.}}
\titlerunning{Social Influence as a Voting System} 

\author{Xavier Molinero\inst{1} \and Fabi\'an Riquelme\inst{2}
\and Maria Serna\inst{2}}
\institute{Dept. of Applied Mathematics III. UPC, Manresa, Spain.\\
\and Dept. de Llenguatges i Sistemes Inform\`atics. UPC, Barcelona, Spain. \\
\email{xavier.molinero@upc.edu, $\{$farisori,mjserna$\}$@lsi.upc.edu}
}
\date{}

\maketitle

\newcommand{\fa}{\forall}
\newcommand{\ex}{\exists}
\newcommand{\bac}{\backslash}

\newcommand{\seb}{\subseteq}
\newcommand{\sepq}{\supseteq}
\newcommand{\sqb}{\sqsubseteq}
\newcommand{\sqp}{\sqsupseteq}
\newcommand{\snb}{\sqsubset}
\newcommand{\snp}{\sqsupset}

\newcommand{\matN}{\mathbb N}
\newcommand{\matR}{\mathbb R}
\newcommand{\matS}{\mathbb S}
\newcommand{\matT}{\mathbb T}
\newcommand{\matH}{\mathbb H}
\newcommand{\matK}{\mathbb K}
\newcommand{\matI}{\mathbb I}
\newcommand{\matO}{\mathbb O}
\newcommand{\matZ}{\mathbb Z}

\newcommand{\cG}{{\cal G}}
\newcommand{\cP}{{\cal P}}
\newcommand{\cH}{{\cal H}}
\newcommand{\cJ}{{\cal J}}
\newcommand{\cK}{{\cal K}}
\newcommand{\cL}{{\cal L}}
\newcommand{\cW}{{\cal W}}
\newcommand{\cZ}{{\cal Z}}
\newcommand{\bG}{\bar{G}}
\renewcommand{\algorithmicrequire}{{\bf Input:}}
\renewcommand{\algorithmicensure}{{\bf Output:}}

\newcommand{\COMMENT}[1]{}
\newcommand{\card}[1]{|{#1}|}
\newtheorem{observation}{Observation}

\begin{abstract}
We consider a simple and altruistic multiagent system in which the agents are eager to perform a collective task but where their real engagement depends on  the willingness to perform the task of other influential agents. We model this scenario by an {\em influence game}, a cooperative simple game in which a team (or coalition) of  players succeeds if it is able to convince enough agents  to participate in the task (to vote in favor of a decision). We take the linear threshold model as the influence model. We show first the expressiveness of influence games showing that they capture the class of simple games. Then we characterize the computational complexity of various problems on influence games, including measures (length and width), values (Shapley-Shubik and Banzhaf) and properties (of teams and players). Finally, we analyze those problems for some particular extremal cases, with respect to the propagation of influence, showing tighter complexity characterizations.
\keywords{Spread of influence, Simple games, Influence games, Computational complexity}
\end{abstract}

\section{Introduction}\label{sec:intro}

Cooperation towards task execution when tasks cannot be performed by a single agent is one of the fundamental problems in both social and multiagent systems. There has been a lot of research understanding collective tasks allocation under different models coming from cooperative game theory. Under such framework, in general, cooperation is achieved by splitting the agents into teams so that each team performs a particular task and the pay-off of the team is split among the team members. Thus, cooperative game theory provides the fundamental tools to analyze this context. Among the many references we refer the reader to \cite{WD04,WD06,CEW11,MF11,BPR13,DNPS14}.

The ways in which people influence each other through their interactions in a social network has received a lot of attention in the last decade.
Social networks have become a huge interdisciplinary research area with important links to sociology, economics, epidemiology, computer science, and mathematics~\cite{Jac08,EK10,AM11,HS14} (players face the choice of adopting a specific product or not; users choose among competing programs from providers of mobile telephones, having the option to adopt more than one product at an extra cost; etc.).
A social network can be represented by a graph where each node is an agent  and each edge represents the degree of influence of one agent over another one. Several ``germs'' (ideas, trends, fashions, ambitions, rules, etc.) can be initiated by one or more agents and eventually be adopted by the system.
The mechanism defining how these motivations are propagated within the network, from the influence of a small set of initially {\em motivated} nodes, is called a model for {\em influence spread}.

Motivated by viral marketing and other applications the problem that has been usually studied is the {\em influence maximization problem} initially introduced by Domingos and Richardson \cite{DR01,RD02} and further developed in \cite{KKT03,EA11}. This problem addresses the question of finding a set with at most $k$ players having maximum influence, and it is {\sc NP}-hard \cite{DR01}, unless additional restrictions are considered, in which case some generality of the problem is lost \cite{RD02}.
Two general models for spread of influence were defined in~\cite{KKT03}: the {\em linear threshold model}, based in the first ideas of \cite{Gra78,Sch78}, and the {\em independent cascade model}, created in the context of marketing by \cite{GLM01,GLM01b}.
Models for influence spread in the presence of multiple competing products has also been proposed and analyzed \cite{BKS07,BFO10,AM11}. In such a setting there is also work done towards analyzing the problem from the point of view of non-cooperative game theory. Non-cooperative {\em influence games} were defined in 2011 by Irfan and Ortiz~\cite{IO11}. Those games, however, analyze the strategic aspects of two firms competing on the social network and differ from our proposal.

We propose to analyze cooperation in multiagent systems based on a model for influence among the agents in their established network of trust and influence. Social influence is relevant to determine the global behavior of a social network and thus it can be used to enforce cooperation by targeting an adequate initial set of agents. From this point of view we consider a simple and altruistic multiagent system in which the agents are eager to perform a collective task but where their real engagement depends on the perception of the willingness to perform the task of other influential agents. We model the scenario by an {\em influence game}, a cooperative simple game in which a team of players (or coalition) succeeds if it is able to convince enough number of agents to participate in the task. We take the deterministic linear threshold model \cite{Che09,AM11} as the mechanism for influence spread in the subjacent social network.

In the considered scenario we adopt the natural point of view of decision or voting systems, mathematically modeled as {\em simple games} \cite{vNM44}.
Simple games were firstly introduced in 1944 by von Neumann and Morgenstern \cite{vNM44} as a fundamental model for social choice. This point of view brings into the analysis several parameters and properties that play a relevant role in the study of simple games and thus in the analysis of the proposed scenario. Among those we consider the {\em length} and the  {\em width}, two fundamental parameters that are indicators of efficiency for making a decision \cite{Ram90}, or the Shapley-Shubik value ({\sc SSval}) and the Banzhaf value ({\sc Bval}) that provide a measure of the notion of individual influence. The properties defining {\em proper}, {\em strong} and {\em decisive} games have been considered in the context of simple game theory from its origins~\cite{TZ99} and they are also studied. Besides those properties we also consider {\em equivalence} and {\em isomorphism}. Together with properties of the games there are several properties associated to players that are of interest. Among others we consider the critical players which were used at least since 1965 by Banzhaf \cite{Ban65}. 
We refer the reader to \cite{TZ99} for a more complete motivation in a viewpoint of simple games and to \cite{Azi09,CEW11} for computational aspects of simple games and in general of cooperative game theory.

To define an influence game we take the spread of influence,  in the linear threshold model, as the value that measures the power of a team. An {\em influence game} is described by an  influence graph, modeling a social network, and a quota, indicating the required minimum number of agents that have to cooperate to perform successfully the task. Therefore, a team will be successful or winning if it can influence at least as many individuals as the quota establishes. Such an approach reveals the importance of the influence between some players over others in order to form successful teams. In this first analysis, we draw upon the deterministic version of the linear threshold model, in which node thresholds are fixed, as our model for influence spread following \cite{Che09,AM11}. It will be of interest to analyze influence games under other spreading models, in particular in the linear threshold model with random thresholds.

Our first result concerns the expressiveness of the family of influence games. We show that unweighted influence games capture the complete family of simple games. Although the construction can be computed in polynomial time when the simple game is given in extensive winning or minimal winning form, the number of winning or minimal winning coalitions is, in general, exponential in the number of players. Interestingly enough the formalization as weighted influence games allows a polynomial time implementation of the operations of intersection and union of weighted simple games, thus showing that, in several cases, simple games that do not admit a succinct representation as weighted games can be represented succinctly as influence games, because its (co)dimension is small. In our characterization we make use of a parameter, the minimum size $k$ for which all coalitions with $k$ members are winning, that turns out to be useful to show that the width of a simple game given in extensive winning or minimal winning form can be computed in polynomial time. This settles an open problem from \cite{Azi09}.

Our second set of results settles the complexity of the problems related to parameters and properties. Hardness results are obtained for unweighted influence games in which the number of agents in the network is polynomial in the number of players, while polynomial time algorithms are devised for general influence games. The new results are summarized in Table \ref{tab:results} as well as the known ones.

\begin{table}[ht]
\begin{center}
\begin{tabular}{|@{\,}l@{\,}|@{\,}c@{\,}|@{\,}c@{\,}|@{\,}c@{\,}|@{\,}c@{\,}|@{\,}c@{\,}|}\hline

& Extensive & Minimal & Weighted  & Multiple & Influence\\
& winning & winning & Games & weighted & games \\
& $(N,\cW)$     & $(N,\cW^m)$       & $[q;w_1,\dots,w_n]$  & games                                 &    $(G,w,f,q,N)$    \\\hline
{\sc Length}       & P             & P                 & P\cite{Azi09}            & NPH\cite{Azi09}     & {\bf NPH}   \\
{\sc Width}        & {\bf P}       & {\bf P}           & P\cite{Azi09}           & P\cite{Azi09}       & {\bf NPH}   \\\hline
{\sc Bval}         & P\cite{Azi08} & \#PC\cite{Azi08}  & \#PC\cite{DP94}    & \#PC                & {\bf \#PC}  \\
{\sc SSval}        & P\cite{Azi08} & \#PC\cite{Azi08}  & \#PC\cite{MM00}  & \#PC                & {\bf \#PC}  \\\hline
{\sc IsDummy}      & P             & P                 & coNPC\cite{MM00}    & coNPC\cite{MM00}    & {\bf coNPC} \\
{\sc IsPasser}     & P\cite{Azi08} & P\cite{Azi08} & P\cite{Azi08}           & P\cite{Azi08}       & {\bf P}     \\
{\sc IsVetoer}     & P\cite{Azi08} & P\cite{Azi08}     & P\cite{Azi08}       & P\cite{Azi08}       & {\bf P}     \\
{\sc IsDictator}   & P\cite{Azi08} & P\cite{Azi08}     & P\cite{Azi08}       & P\cite{Azi08}       & {\bf P}     \\
{\sc AreSymmetric} & P             & P                 & coNPC\cite{MM00}    & coNPC\cite{MM00}    & {\bf P}     \\\hline
{\sc IsCritical}   & P             & P                 & P                   & P                   & {\bf P}     \\
{\sc IsBlocking}   & P             & P                 & P                   & P                   & {\bf P}     \\
{\sc IsProper}     & P             & P                 & coNPC\cite{FMOS12}  & coNPC               & {\bf coNPC} \\
{\sc IsStrong}     & P             & coNPC\cite{PCPO02}& coNPC\cite{FMOS12}  & coNPC               & {\bf coNPC} \\
{\sc IsDecisive}   & P             & QP\cite{FK96}     & coNPC\cite{Azi09}   & coNPC               & {\bf coNPC} \\\hline
{\sc Equiv}        & P             & P                 & coNPC\cite{EGGW08}& coNPH\cite{EGGW08}& {\bf coNPH} \\
{\sc Iso}          & {\sc gIso}    & {\sc gIso}        & ?                   & ?                   & {\bf coNPH} \\\hline
\end{tabular}
\caption{Summary of new results (in bold face), known results  and trivial results (without reference).\label{tab:results}}
\end{center}
\end{table}

We refer the reader to Sections~\ref{sec:defs} and \ref{sec:par-pro} for a formal definition of all the representations mentioned in the first row and the problems in the first column of  Table \ref{tab:results}. There {\sc P} (polynomial time solvable), {\sc \#PC} (\#P-complete), {\sc NPH} (NP-hard), {\sc coNPH} (coNP-hard), {\sc coNPC} (coNP-complete),  {\sc QP} (quasi-polynomial time solvable) and  {\sc gIso} (the class of problems reducible to  graph isomorphism) are known computational complexity classes \cite{GJ79,Pap94}. The isomorphism problems for simple games, given either by $(N,\cW)$ or  $(N,\cW^m)$, is easily shown to be polynomially reducible to the graph isomorphism problem.  For  games given by $(N,\cW^m)$,  the {\sc Iso} problem and the graph isomorphism problem are equivalent using arguments from~\cite{Luk99}. 

Finally, we consider two extreme cases of influence spread in social networks for undirected and unweighted influence games. In a {\em maximum influence requirement}, agents adopt a behavior only when all its peers have already adopted it. This is opposed to a {\em minimum influence requirement} in which an agent gets convinced when at least one of its peers does. We show that, in both cases, the problems {\sc IsProper}, {\sc IsStrong} and {\sc IsDecisive}, as well as computing {\sc Width}, have polynomial time algorithms. Computing {\sc Length} is {\sc NP}-hard for maximum influence and polynomial time solvable for minimum influence.

\section{Definitions and Preliminaries}\label{sec:defs}

Before introducing formally the family of influence games  we need to define a family of labeled graphs and a process of spread of influence based on the {\em linear threshold model}~\cite{Gra78,Sch78}. We use standard graph notation following~\cite{Bol98}.  As usual, given a finite set $U$, $\cP(U)$ denotes its power set, and $|U|$ its cardinality. For any  $0\leq k\leq |U|$,  $\cP_k(U)$ denotes the subsets of $U$ with exactly $k$-elements.  For a given graph $G=(V,E)$ we assume that $n=|V|$ and $m=|E|$. Also $G[S]$ denotes the subgraph induced by $S\subseteq V$.

\begin{definition}\label{def:influence_graph}
An {\em influence graph} is a tuple $(G,w,f)$, where $G=(V,E)$ is a weighted, labeled and directed graph (without loops).
As usual $V$ is the set of vertices or agents, $E$ is the set of edges and $w:E\to\matN$ is a {\em weight function}. Finally,  $f:V\to\matN$ is a labeling function that quantifies how influenceable each agent is. An agent $i\in V$ has {\em influence} over another $j\in V$ if and only if $(i,j)\in E$. 
We also consider the family of {\em unweighted influence graphs} $(G, f)$ in which every edge has weight 1.
\end{definition}

Given an influence graph $(G,w,f)$ and an initial activation set $X\seb V$, the {\em spread of influence} of $X$ is the set  $F(X)\seb V$ which is
formed by the agents activated through an iterative process.  We use $F_k(X)$ to denote the set of nodes activated at step $k$.  Initially, at step 0, only the vertices in $X$ are activated, that is $F_0(X)=X$.
At  step $i>0$,  those vertices for which the sum of weights of the edges connecting  nodes in $F_{i-1}(X)$ to them meets or exceeds their label functions are activated, i.e.,
$$F_i(X) = F_{i-1}(X) \cup \{v\in V\mid \textstyle\sum_{\{u\in F_{i-1}(X)\mid(u,v)\in E\}} w((u,v))\geq f(v)\}.$$
The process stops when no additional activation occurs and the final set of activated nodes becomes $F(X)$.

\begin{example}\label{ex1}
Figure \ref{fig1} shows the spread of influence $F(X)$ in an unweighted influence graph $G=(V,f)$, with $V=\{a,b,c,d\}$, for the initial activation $X=\{a\}$.
In the first step we obtain $F_1(X)=\{a,c\}$, and in the second step (the last one) we obtain $F(X)=F_2(x)=\{a,c,d\}$.

\begin{figure}[t]
\centering
\begin{minipage}[b]{0.25\linewidth}
\begin{center}
\begin{tikzpicture}[every node/.style={circle,scale=0.8}, node distance=7mm, >=latex]
\node[scale=1] at (1,2.5){$F_0(X)=X=\{a\}$.};
\node[fill=blue!20](a) at (0,2)[label=left:$a$] {1};
\node[draw](b) at (2,2)[label=right:$b$] {1};
\node[draw](c) at (0,0)[label=left:$c$] {1};
\node[draw](d) at (2,0)[label=right:$d$] {2};
\draw[->] (a) to node {}(c);
\draw[->] (a) to node {}(d);
\draw[->] (b) to node {}(a);
\draw[->] (b) to node {}(d);
\draw[->] (c) to node {}(d);
\end{tikzpicture}
\end{center}
\end{minipage}
\qquad
\begin{minipage}[b]{0.25\linewidth}
\begin{center}
\begin{tikzpicture}[every node/.style={circle,scale=0.8}, node distance=7mm, >=latex]
\node[scale=1] at (1,2.5){$F_1(X)=\{a,c\}$.};
\node[fill=blue!20](a) at (0,2)[label=left:$a$] {1};
\node[draw](b) at (2,2)[label=right:$b$] {1};
\node[fill=blue!20](c)[label=left:$c$] at (0,0) {1};
\node[draw](d)[label=right:$d$] at (2,0) {2};
\draw[->] (a) to node {}(c);
\draw[->] (a) to node {}(d);
\draw[->] (b) to node {}(a);
\draw[->] (b) to node {}(d);
\draw[->] (c) to node {}(d);
\end{tikzpicture}
\end{center}
\end{minipage}
\qquad
\begin{minipage}[b]{0.25\linewidth}
\begin{center}
\begin{tikzpicture}[every node/.style={circle,scale=0.8}, node distance=7mm, >=latex]
\node[scale=1] at (1,2.5){$F(X)=F_2(X)=\{a,c,d\}$.};
\node[fill=blue!20](a) at (0,2)[label=left:$a$] {1};
\node[draw](b) at (2,2)[label=right:$b$] {1};
\node[fill=blue!20](c) at (0,0)[label=left:$c$] {1};
\node[fill=blue!20](d) at (2,0)[label=right:$d$] {2};
\draw[->] (a) to node {}(c);
\draw[->] (a) to node {}(d);
\draw[->] (b) to node {}(a);
\draw[->] (b) to node {}(d);
\draw[->] (c) to node {}(d);
\end{tikzpicture}
\end{center}
\end{minipage}
\caption{The spread of influence starting from the initial activation of $X=\{a\}$ on an unweighted graph.\label{fig1}}
\end{figure}
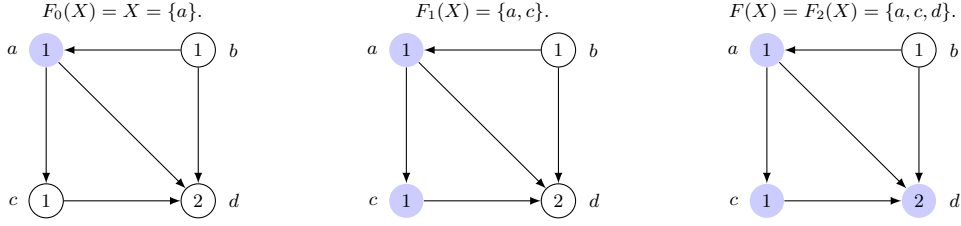
\end{example}

 As the number of vertices is finite,  for any $i>n$, $F_i(X)=F_{i-1}(X)$. Thus, $F(X)=F_n(X)$ and we have the following well known basic result.

\begin{lemma}\label{lem:Fpoly}
Given an influence graph $(G,w,f)$ and a set of vertices $X$, the set $F(X)$ can be computed in polynomial time.
\end{lemma}

In what follows, unless otherwise stated, results and definitions will be stated for directed graphs. All of them can be restated for undirected graphs. Now we define {influence games}.

\begin{definition}\label{def:influence_game}
An {\em influence game} is given by a tuple $(G,w,f,q,N)$ where $(G,w,f)$ is an influence graph, $q$ is an integer {\em quota}, $0\leq q\leq |V|+1$, and $N\seb V$ is the {\em set of players}. $X\seb N$ is a {\em successful team} if and only if $|F(X)|\geq q$, otherwise $X$ is an {\em unsuccessful team}.
\end{definition}

As it was done for influence graphs, we also consider the family of {\em unweighted influence games} for the cases in which the graph $G$ is unweighted. In such a case we use the notation $(G,f,q,N)$.

An influence game adopts a correspondence with a mathematical model called simple game, where teams are coalitions, and some agents  act as  players.
Simple games were firstly introduced in 1944 by Neumann and Morgenstern~\cite{vNM44}, but using the definition that corresponds to the so called {\em strong games}, which will be defined later. The first definition as we understand simple games was given in 1953 by Gillies~\cite{Gil53}.
In this scenario of simple games, we follow definitions and notation from \cite{TZ99}.
A family of subsets $\cW\seb\cP(N)$ is said {\em monotonic} when, for any  $X\in\cW$ and $Z\in\cP(N)$,  if $X\seb Z$ , then $Z\in\cW$.

\begin{definition}\label{def:SG}
A {\em simple game} $\Gamma$ is given by a tuple $(N,\cW)$ where $N$ is a finite set of players and $\cW$ is a monotonic family of subsets of $N$ formed by the {\em winning coalitions} ({\em successful teams}).
\end{definition}

In the context of simple games, the subsets of $N$ are called {\em coalitions}, $N$ is the {\em grand coalition} and $X\in\cW$ is a {\em winning coalition}.
Any subset of $N$ which is not a winning coalition is called  a {\em losing coalition} (an {\em unsuccessful team}).
A {\em minimal winning coalition} is a winning coalition $X$ that does not properly contain  any winning coalition. That is, removing any player from $X$ results in a losing coalition. A {\em maximal losing coalition} is a losing coalition $X$ that is not properly contained in  any other losing coalition. That is, adding any  player to $X$ results in a winning coalition. We use $\cW$, $\cL$, $\cW^m$ and $\cL^M$ to denote the sets of winning, losing, minimal winning and maximal losing coalitions, respectively. Any of those set families determine uniquely the game and constitute the usual forms of representation for simple games~\cite{TZ99}, although the size of those representations are not, in general, polynomial in the number of players.

\begin{example}
Let $(G,f)$ be an influence graph  and $N$ any subset of agents.  Two particular ranges of the quota lead to some trivial simple games. 
By setting $q=0$, thus considering influence games of the form $(G,f,0,N)$,  we have that every team of agents is  successful, therefore $(G,f,0,N)$ is a representation of the simple game $(N,\cP(N))$. 
When $q>|V(G)|$, the influence game $(G,f,q,N)$ is a representation of the simple game $(N,\emptyset)$  as there are no successful teams in the game.
\end{example}

Let us provide an example of influence game based on the influence graph considered in Example \ref{ex1}.
\begin{example}
Consider the influence game $\Gamma=(G,f,3,V(G))$, where $(G,f)$ is the influence graph considered in Example \ref{ex1}. In this case, we have that, $F(\{a\})=\{a,c,d\}$, and thus $\{a\}\in\cW$. 
The fundamental set families for $\Gamma$ are:
\[
\cW^m   =  \{\{a\},\{b\}\} \quad
\cL^M     =  \{\{c,d\}\} \quad
\cL         = \{ \{c,d\}, \{c\}, \{d\}, \{\}\}\quad
\cW      = \cP(V(G))\setminus \cL.
\]
\end{example}

The {\em intersection} of two simple games is the simple game where a coalition wins if and only if it wins in both games.
In a similar way, the {\em union} of two simple games is the simple game where a coalition wins if and only if it wins in at least one of the two games~\cite{TZ99}.  

Finally,  we introduce a subfamily of simple games, the \emph{weighted games}. 
\begin{definition}\label{def:WG}
A simple game $(N,\cW)$ is a {\em weighted game} (also called {\em weighted voting game}) if there exists a {\em weight function} $w:N\to\matR^+$ and a {\em quota} $q\in\matR$, such that, for any $ X\seb N$, $X\in\cW$ if and only if $w(X)\geq q$, where $w(X)=\sum_{i\in X}w(i)$.
\end{definition}

A weighted game can be represented by a {\em weighted representation}, i.e., a vector $[q;w_1,\ldots,w_n]$ where $w_i=w(i)$, for any $i\in N$, and $q\in\matR$ is a quota, defining the simple game in which $S\in\cW$ if and only if $w(S)\geq q$.
According to Hu~\cite{Hu65} (see also~\cite{FM09}) the weighted representations can be restricted to integer non-negative weights with  $0\leq q\leq w(N)$.
{\em Multiple weighted games} are simple games defined by the intersection of a finite collection of weighted games.

Despite the fact that weighted games are a strict subclass of simple games, it is known that every simple game can be expressed as an intersection or an union of a finite number of weighted games. The result for intersection ({\em dimension} concept) was firstly shown in~\cite{Jer75} for hypergraphs, and then expressed for simple games in~\cite{TZ93}. The result for union ({\em codimension} concept) was introduced for simple games in~\cite{FM09a}.
A simple game is said to be of {\em dimension} ({\em codimension}) $k$ if and only if it
can be represented as the intersection (union) of exactly $k$ weighted games, but not
as the intersection (union) of $(k - 1)$ weighted games. It is known that given $k$ weighted games, to decide whether the dimension of their intersection exactly equals $k$ is {\sc NP}-hard~\cite{DW06}.
A generalization of games constructed through binary operators is the family of \emph{boolean weighted}  games  introduced in \cite{FEW09}.  A boolean weighted game is defined by a propositional logic formula  and a finite collection of weighted games. The boolean formula determines the requirements for a coalition to be winning in the described game. When considering only monotone formulas, boolean weighted games provide another representation of simple games.  

\section{Expressiveness}\label{sec:IG}

Influence games are monotonic as, for any $X\seb N$ and  $i\in N$, if $|F(X)|\geq q$ then $|F(X\cup\{i\})|\geq q$, and if $|F(X)|<q$ then $|F(X\bac\{i\})|<q$. Thus, every influence game is a simple game. Moreover, we will show that the opposite is also true. Before stating the main theorem we need the definition of a new measure over simple games.

\begin{definition}
The {\em strict length} of a simple game $\Gamma$ is $\mbox{\sc sLength}(\Gamma)=\min\{k\in\mathbb{N}\mid  \cP_k(N) \subseteq \cW\}$.
\end{definition}

The measure {\sc sLength} will be considered later in Section~\ref{sec:par-pro}  together with other  measures for simple games.

\begin{theorem}\label{the:SG-Inf}\label{lem:slength-min-win}
Every simple game can be represented by an unweighted influence game.
Furthermore, when the simple game $\Gamma$  is given by either $(N,\cW)$ or $(N,\cW^m)$, an  unweighted influence game representing $\Gamma$  can be obtained in polynomial time.
\end{theorem}

\begin{proof}
Assume that a simple game $\Gamma$ is given by $(N,\cW)$ or $(N,\cW^m)$. It is already well known that given  $(N,\cW)$, the family $\cW^m$ can be obtained in polynomial time. Thus we assume in the following that the set of players and the set $\cW^m$ are given.

In order to represent $\Gamma$ as an influence game  we first define  an unweighted influence graph $(G,f)$.  The graph $G=(V,E)$  is the following.
The set $V$ of nodes is formed by a set with $n$ nodes, $V_N=\{v_1,\dots,v_n\}$, one for each player, and a set of nodes for each minimal winning coalition.  For any  $ X\in\cW^m$, we add  a new set  $V_X$ with  $\mbox{\sc sLength}(\Gamma)-|X|$ nodes. We connect vertex $v_i$ with all the vertices in $V_X$ whenever $i\in X$.  Finally, the label function is  defined as follows, for any $1\leq i\leq n$, $f(v_i)=1$ and, for any $X\in\cW^m$ and any $v\in V_X$, $f(v)=|X|$.
 Observe that in the influence game $(G,f,\mbox{\sc sLength}(\Gamma),V_N)$ a team is successful if and only if its players form a  winning coalition in $\Gamma$. Therefore $(G,f,\mbox{\sc sLength}(\Gamma),V_N)$ is a representation of $\Gamma$ as  unweighted influence game. 

It remains to show that given $(N,\cW^m)$ a description of  $(G,f,\mbox{\sc sLength}(\Gamma),N)$ can be computed in polynomial time. For doing so it is enough to show that $\mbox{\sc sLength}(\Gamma)$ can be computed in polynomial time. Let $k=\mbox{\sc sLength}(\Gamma)$.

Observe that, by definition,  all the coalitions  with  $k$ players are winning in $\Gamma$ but  at least one coalition with size  $k-1$ is losing. Therefore there is a minimal winning coalition with size $k$ and there are no minimal winning coalitions with size $k+1$. Thus, computing $k$ is equivalent to compute the maximum size of a minimal winning coalition. The last quantity can be obtained in polynomial time from a description of $\cW^m$. 
\end{proof}

The following example provides an illustration of the construction.

\begin{example}
Let $\Gamma=(\{1,2,3,4\},\{\{1,2,4\},\{2,3\},\{3,4\}\})$ be a simple game in minimal winning form.
We have that $\mbox{\sc sLength}(\Gamma)=3$ because all subsets of $N$ with cardinality $3$ are winning, i.e., we have that the family $\{1,2,3\},\{1,2,4\},\{1,3,4\},\{2,3,4\}\in\cW$.
For coalition $\{1,2,4\}$ we do not need to add nodes to the graph. For each of the teams $\{2,3\}$ and $\{3,4\}$, we need to add one node with label $3-2=1$. A drawing of the resulting unweighted influence graph is given in Figure~\ref{fig:Ex2}.

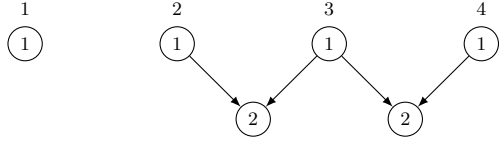
\begin{figure}[t]
\centering
\begin{tikzpicture}[every node/.style={circle,scale=0.8}, >=latex]
\node[draw](a) at (1,1)[label=above:$1$] {1};
\node[draw](b) at (3,1)[label=above:$2$] {1};
\node[draw](c) at (5,1)[label=above:$3$] {1};
\node[draw](d) at (7,1)[label=above:$4$] {1};

\node[draw](g) at (4,0) {2};
\node[draw](h) at (6,0) {2};

\draw[->] (b) to node {}(g);
\draw[->] (c) to node {}(g);
\draw[->] (c) to node {}(h);
\draw[->] (d) to node {}(h);
\end{tikzpicture}
\caption{An unweighted influence graph associated to  the simple game $(\{1,2,3,4\},\{\{1,2,4\},\{2,3\},\{3,4\}\})$.\label{fig:Ex2}}
\end{figure}
\end{example}

The proof of Theorem \ref{the:SG-Inf} shows the expressiveness of the family of influence games with respect to the class of simple games.
However, the construction cannot be implemented in polynomial time when the simple game is given in succinct way, as for instance by a weighted representation or as a monotonic Boolean function. Observe also that the number of agents in the corresponding influence game is in general exponential in the number of players. 
For the particular case of weighted  games,  with given weighted representation, we can show that there exist representations by influence games having a polynomial number of agents.

\begin{theorem}\label{the:WG-Inf}
Every weighted game can be represented as an influence game. Furthermore, given a weighted representation of the game, a representation as an influence game can be obtained in polynomial time.
\end{theorem}

\begin{figure}[t]
\centering
\begin{tikzpicture}[every node/.style={circle,scale=0.6}, >=latex]
\node[draw](a)   at (3,2)  {1};
\node[scale=2]   at (4,2){$\ldots$};
\node[draw](c)   at (5,2)  {1};

\node[draw,scale=1.5](m)   at (4,1) {$q$};

\draw[->] (a) to node[label=left:{\Large $w_1$\ \ }] {}(m);
\draw[->] (c) to node[label=right:{\ \ \ \Large $w_n$}] {}(m);

\node[scale=2,left] at (7,2){\small{$n$ nodes}};
\node[scale=2,left] at (7,0){\small{$n$ nodes}};
\node[scale=2,left] at (7,-1){$\ $};

\node[draw](x)   at (3,0)  {1};
\node[scale=2]   at (4,0){$\ldots$};
\node[draw](z)   at (5,0)  {1};

\draw[->] (m) to node[label=left:{\Large $1$\ \ \ }] {}(x);
\draw[->] (m) to node[label=right:{\ \ \ \Large $1$}] {}(z);
\end{tikzpicture}
\vspace{-5ex}\caption{An influence graph $(G,w,f)$ associated to the weighted game $[q;w_1,\ldots,w_n]$.\label{fig:Wei-Inf}}
\end{figure}
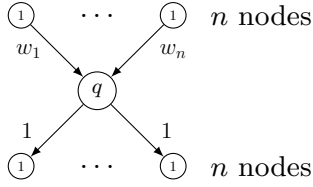

\begin{proof}
Let $[q;w_1,\ldots,w_n]$ be a weighted game, consider the influence game $(G,w,f,n+1,N)$, whose influence graph is shown in Figure~\ref{fig:Wei-Inf}.
The $n$ nodes in the first level of $G$ correspond to the set $N$, each of them has as associated label the value $1$.
Each node $i\in N$ is connected to a central  node with label $q$ and  the corresponding edge has weight $w_i$.
The $n$ nodes in the last level are another set of $n$ nodes with label $1$.
Observe that, $X\seb N$ is a winning coalition in  $[q;w_1,\ldots,w_n]$ if and only if $\sum_{i\in X}w_i\geq q$. The last condition is equivalent to  $|F(X)|\ge n+1$.Thus we have that  $X\seb N$ is a winning coalition in  $[q;w_1,\ldots,w_n]$ if and only if  $X\seb N$ is a winning coalition in $(G,w,f,n+1,N)$. 

Finally, observe that the construction of $(G,w,f,n+1,N)$ can be done in polynomial time with respect to the size of  $[q;w_1,\ldots,w_n]$.
\end{proof}

Observe that in the previous construction the size of the influence graph is polynomial in the number of agents but the overall construction is done in polynomial time in the size of the  weighted representation.  We can change slightly the construction and get a representation as unweighted influence game by increasing again the proportion of players. 

\begin{theorem}\label{the:WG-UIG}
Every weighted game can be represented as an unweighted influence game.
Furthermore, given a weighted representation of the game, a representation as  unweighted influence game can be obtained in pseudo-polynomial time.
\end{theorem}

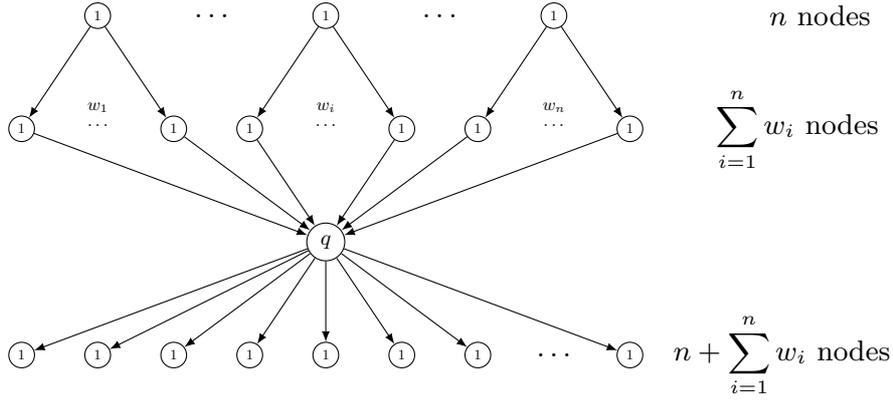
\begin{figure}[t]
\centering
\begin{tikzpicture}[every node/.style={circle,scale=0.6}, >=latex]
\node[draw](a)   at (1,4.5)  {1};
\node[scale=2]   at (2.5,4.5){$\ldots$};
\node[draw](b)   at (4,4.5)  {1};
\node[scale=2]   at (5.5,4.5){$\ldots$};
\node[draw](c)   at (7,4.5)  {1};

\node[draw](d)   at (0,3)  {1};
\node[scale=1.2] at (1,3.2){{$\begin{array}{c}w_1\\[-.1cm] \ldots\end{array}$}};
\node[draw](g)   at (2,3)  {1};
\node[draw](h)   at (3,3)  {1};
\node[scale=1.2] at (4,3.2){{$\begin{array}{c}w_i\\[-.1cm] \ldots\end{array}$}};
\node[draw](j)   at (5,3)  {1};
\node[draw](k)   at (6,3)  {1};
\node[scale=1.2] at (7,3.2){{$\begin{array}{c}w_n\\[-.1cm] \ldots\end{array}$}};
\node[draw](l)   at (8,3)  {1};

\node[draw,scale=1.5](m) at (4,1.5) {$q$};

\node[draw](n)   at (0,0)  {1};
\node[draw](o)   at (1,0)  {1};
\node[draw](p)   at (2,0)  {1};
\node[draw](q)   at (3,0)  {1};
\node[draw](r)   at (4,0)  {1};
\node[draw](s)   at (5,0)  {1};
\node[draw](t)   at (6,0)  {1};
\node[scale=2]   at (7,0)  {$\ldots$};
\node[draw](u)   at (8,0)  {1};

\draw[->] (a) to node {}(d);
\draw[->] (a) to node {}(g);
\draw[->] (b) to node {}(h);
\draw[->] (b) to node {}(j);
\draw[->] (c) to node {}(k);
\draw[->] (c) to node {}(l);
\draw[->] (d) to node {}(m);
\draw[->] (g) to node {}(m);
\draw[->] (h) to node {}(m);
\draw[->] (j) to node {}(m);
\draw[->] (k) to node {}(m);
\draw[->] (l) to node {}(m);
\draw[->] (m) to node {}(n);
\draw[->] (m) to node {}(o);
\draw[->] (m) to node {}(p);
\draw[->] (m) to node {}(q);
\draw[->] (m) to node {}(r);
\draw[->] (m) to node {}(s);
\draw[->] (m) to node {}(t);
\draw[->] (m) to node {}(u);

\node[scale=2] at (10.5,4.5){\small{$n$ nodes}};
\node[scale=2] at (10.2,3){\small{$\displaystyle\sum^n_{i=1}w_i$ nodes}};
\node[scale=2] at (10,0){\small{$n+\displaystyle\sum^n_{i=1}w_i$ nodes}};
\end{tikzpicture}
\vspace{-5ex}\caption{An unweighted influence graph $(G,f)$ associated to  the weighted game $[q;w_1,\ldots,w_n]$.\label{fig:WG-Inf}}
\end{figure}

\begin{proof}
Let $[q;w_1,\ldots,w_n]$ be a weighted game, consider the unweighted influence graph $(G,f)$ sketched in Figure~\ref{fig:WG-Inf}.
The $n$ nodes in the first level correspond to the set $N$.
For any $i\in N$, node $i$ is connected to a set of $w_i$ different nodes in the second level representing its weight.
Thus,  $X\seb N$ is a winning coalition if and only if $\sum_{i\in X}w_i\geq q$, which is equivalent to $|F(X)|\ge{}n+\sum_{i=1}^n w_i$.
Therefore, the influence game $(G,f,n+\sum_{i=1}^n w_i,N)$ is a representation of   the given weighted game.

Observe that given $[q;w_1,\ldots,w_n]$, constructing the graph $G$ requires time $O(n+ w_1+\cdots+w_n)$ and thus the construction can be done in pseudo-polynomial time.  
\end{proof}

In the previous results we have assumed that a weighted representation of the game is given. It is  known that there are weighted games whose  weighted representation  requires that $\max_{i\in{}N}\{w_i\}$ to be $(n+1)^{(n+1)/2}/2$~\cite{Par94}. Therefore the construction of the previous lemma will require exponential space and time with respect to the number of players.

Our next result establishes the closure of influence games under intersection and union. Furthermore, we show that an influence game representing the resulting simple game can be obtained in polynomial time.

\begin{theorem}\label{the:SG-UIG}
Given two influence games, their intersection and union  can be represented as an influence game.
Furthermore, both constructions can be obtained in polynomial time.
\end{theorem}

\begin{proof}
Let $\Gamma=(G,w,f,q,N)$ be an influence game with  $G=(V,E)$, recall that, for any $X\seb N$,  $F_i(X)\seb V$ denotes the spread of influence of $X$ in the $i$-th step of the activation process and that  we can assume that  $0\leq i\leq n$.  All the sets considered in our constructions are replications of either the set  $N$ or the set  $V$. For sake of simplicity, we use the term \emph{corresponding node} to refer to the same node in a different copy of $N$ or $V$.

We start constructing  an influence graph $(G',w',f')$ as shown in Figure~\ref{fig:SG-Fi(X)}. $G'$ has $2n+1$ columns of nodes. The first column $F^0$ represents $V$, and the remaining nodes  are divided in pairs of  sets $(f^i,F^i)$, for any $1\leq i \leq n$. 
For any $1\leq i\leq n$, the sets  $f^i$ and $F^i$ have $n$ nodes each, as a replication of  the nodes in $V$. 
The edges are defined as follows, for any $1\leq i\leq n$, a  node $y\in F^{i-1}$  is connected to a node $z\in f^{i}$ if and only if $(y,z)\in E$. These edges have associated weight $w((y,z))$.  Furthermore, every node in $F^{i-1}$ is connected by an edge with weight 1 to its corresponding node in $F^i$.  Every node in $f^{i}$ is connected by an edge with weight 1 to its corresponding  node in $F^i$. The labeling function assigns label 1 to all the nodes in sets $F^i$ and maintains the original labeling for nodes in the sets $f^i$.

Note that after the  activation of a team $X\subseteq F^0$ in $(G',w',f')$, for any $0\leq i\leq n$, the set of nodes in $F^i$ that are activated coincides with the set $F_i(X)$. Thus the subset of activated nodes in $F^n$ coincides with $F(X)$. Observe also that $(G',w',f')$ has $2n^2+n$ nodes and that it can  be constructed in polynomial time in the size of a given influence game $(G,w,f,q,N)$.
 
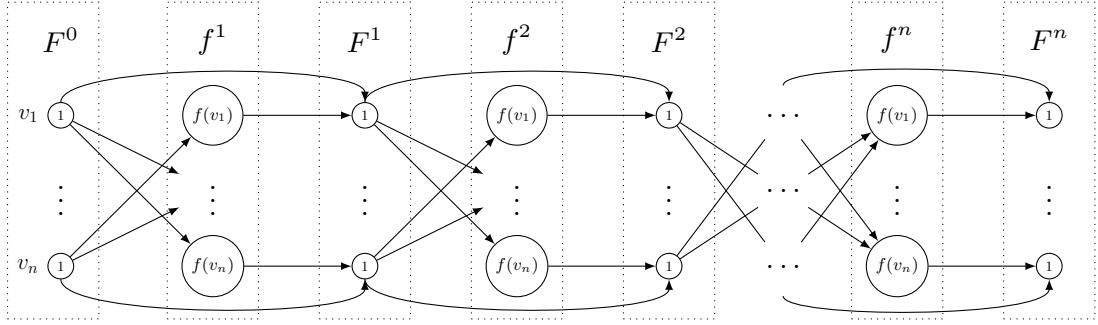
\begin{figure}[t]
\centering
\begin{tikzpicture}[every node/.style={circle,scale=0.6}, >=latex]
\node[scale=2]             at (1,3){$F^0$};
\node[scale=2]             at (3,3){$f^1$};
\node[scale=2]             at (5,3){$F^1$};
\node[scale=2]             at (7,3){$f^2$};
\node[scale=2]             at (9,3){$F^2$};
\node[scale=2]             at (12,3){$f^n$};
\node[scale=2]             at (14,3){$F^n$};

\draw[dotted] (0.3,-.7) rectangle (1.5,3.5);
\node[draw](v1a)           at (1,2)[label=left:\Large{$v_1$}]{1};
\node[scale=2]             at (1,1){$\vdots$};
\node[draw](vna)           at (1,0)[label=left:\Large{$v_n$}]{1};

\draw[dotted] (2.4,-.7) rectangle (3.6,3.5);
\node[scale=1.2,draw](fv1a)at (3,2){$f(v_1)$};
\node[scale=2](fvxa)       at (3,1){$\vdots$};
\node[scale=1.2,draw](fvna)at (3,0){$f(v_n)$};

\draw[dotted] (4.4,-.7) rectangle (5.6,3.5);
\node[draw](v1b)           at (5,2){1};
\node[scale=2]             at (5,1){$\vdots$};
\node[draw](vnb)           at (5,0){1};

\draw[dotted] (6.4,-.7) rectangle (7.6,3.5);
\node[scale=1.2,draw](fv1b)at (7,2){$f(v_1)$};
\node[scale=2](fvxb)       at (7,1){$\vdots$};
\node[scale=1.2,draw](fvnb)at (7,0){$f(v_n)$};

\draw[dotted] (8.4,-.7) rectangle (9.6,3.5);
\node[draw](v1c)           at (9,2){1};
\node[scale=2]             at (9,1){$\vdots$};
\node[draw](vnc)           at (9,0){1};

\node[scale=2](d1)         at (10.5,2){$\ldots$};
\node[scale=2](d2)         at (10.5,1){$\ldots$};
\node[scale=2](d3)         at (10.5,0){$\ldots$};

\draw[dotted] (11.4,-.7) rectangle (12.6,3.5);
\node[scale=1.2,draw](fv1z)at (12,2){$f(v_1)$};
\node[scale=2](fvxz)       at (12,1){$\vdots$};
\node[scale=1.2,draw](fvnz)at (12,0){$f(v_n)$};

\draw[dotted] (13.4,-.7) rectangle (14.6,3.5);
\node[draw](v1z)           at (14,2){1};
\node[scale=2]             at (14,1){$\vdots$};
\node[draw](vnz)           at (14,0){1};

\draw[->](v1a)  to node{}(fvxa);
\draw[->](v1a)  to node{}(fvna);
\draw[->](vna)  to node{}(fv1a);
\draw[->](vna)  to node{}(fvxa);

\draw[->](fv1a) to node{}(v1b);
\draw[->](fvna) to node{}(vnb);
\draw[->](v1a) .. controls +(up:.7cm)   and +(up:.7cm)   .. node[above,sloped]{}(v1b);
\draw[->](vna) .. controls +(down:.7cm) and +(down:.7cm) .. node[above,sloped]{}(vnb);

\draw[->](v1b)  to node{}(fvxb);
\draw[->](v1b)  to node{}(fvnb);
\draw[->](vnb)  to node{}(fv1b);
\draw[->](vnb)  to node{}(fvxb);

\draw[->](fv1b) to node{}(v1c);
\draw[->](fvnb) to node{}(vnc);
\draw[->](v1b) .. controls +(up:.7cm)   and +(up:.7cm)   .. node[above,sloped]{}(v1c);
\draw[->](vnb) .. controls +(down:.7cm) and +(down:.7cm) .. node[above,sloped]{}(vnc);

\draw[-](v1c)   to node{}(d2);
\draw[-](v1c)   to node{}(d3);
\draw[-](vnc)   to node{}(d1);
\draw[-](vnc)   to node{}(d2);

\draw[->](d2)   to node{}(fv1z);
\draw[->](d1)   to node{}(fvnz);
\draw[->](d3)   to node{}(fv1z);
\draw[->](d2)   to node{}(fvnz);

\draw[->](fv1z) to node{}(v1z);
\draw[->](fvnz) to node{}(vnz);
\draw[->](d1) .. controls +(up:.7cm)   and +(up:.7cm)   .. node[above,sloped]{}(v1z);
\draw[->](d3) .. controls +(down:.7cm) and +(down:.7cm) .. node[above,sloped]{}(vnz);
\end{tikzpicture}
\caption{The influence graph $(G',w',f')$ associated to the influence game $(G,w,f,q,N)$.\label{fig:SG-Fi(X)}}
\end{figure}

\begin{figure}[t]
\centering
\begin{tikzpicture}[every node/.style={circle,scale=0.6}, >=latex]

\node[scale=2]             at (1,3){$F$};
\draw[dotted] (2,3.5) rectangle (0,-1);
\node[draw](v1)           at (1,2){1};
\node[scale=2]            at (1,1){$\vdots$};
\node[scale=2]            at (0.5,1){$N$};
\node[draw](vn)           at (1,0){1};

\node[scale=1]            at (2.5,2){$F_1^0$};
\node[scale=1]            at (4.8,2){$F_1^n$};

\node[scale=1]            at (2.5,.1){$F_2^0$};
\node[scale=1]            at (4.8,.1){$F_2^n$};

\draw[dashed] (2.3,1.3) rectangle (5,2.5);
\draw[dotted](2.3,1.3)  rectangle (2.7,2.5);
\draw[dotted](4.6,1.3)  rectangle (5,2.5);
\draw[dashed] (2.3,-.5) rectangle (5,.7);
\draw[dotted](2.3,-.5)  rectangle (2.7,.7);
\draw[dotted](4.6,-.5)  rectangle (5,.7);
\node[scale=2]            at (3.5,3){$(G'_1,w'_1,f'_1)$};
\node[scale=2]            at (3.5,-1){$(G'_2,w'_2,f'_2)$};
\node(x1)                 at (2.5,2.5){};
\node(x2)                 at (2.5,1.5){};
\node(x3)                 at (2.5,0.5){};
\node(x4)                 at (2.5,-.5){};
\node(y1)                 at (4.9,2.5){};
\node(y2)                 at (4.9,1.3){};
\node(y3)                 at (4.9,0.7){};
\node(y4)                 at (4.9,-.5){};

\node[scale=1.2,draw](q1) at (6.5,2){$q_1$};
\node[scale=1.2,draw](q2) at (6.5,0){$q_2$};
\node[scale=1.2,draw](x)  at (8,1)  {$x$};
\node[draw](s1)           at (9.5,2)[label=right:\huge{$s_1$}]{1};
\node[scale=2]            at (9.5,1){$\vdots$};
\node[draw](sn)           at (9.5,0)[label=right:\huge{$s_{4n^2+3n+2}$}]{1};

\draw[->](v1a) to node{}(x1);
\draw[->](v1a) to node{}(x3);
\draw[->](vna) to node{}(x2);
\draw[->](vna) to node{}(x4);

\draw[->](y1)  to node{}(q1);
\draw[->](y2)  to node{}(q1);
\draw[->](y3)  to node{}(q2);
\draw[->](y4)  to node{}(q2);

\draw[->](q1)  to node{}(x);
\draw[->](q2)  to node{}(x);
\draw[->](x)   to node{}(s1);
\draw[->](x)   to node{}(sn);
\end{tikzpicture}
\caption{The influence graph associated to the intersection $(x=2)$ or the union $(x=1)$ of two influence games with influence graphs $(G_1,w_1,f_1)$ and $(G_2,w_2,f_2)$ and quotas $q_1$ and $q_2$ respectively.\label{fig:UI-IGs}}
\end{figure}
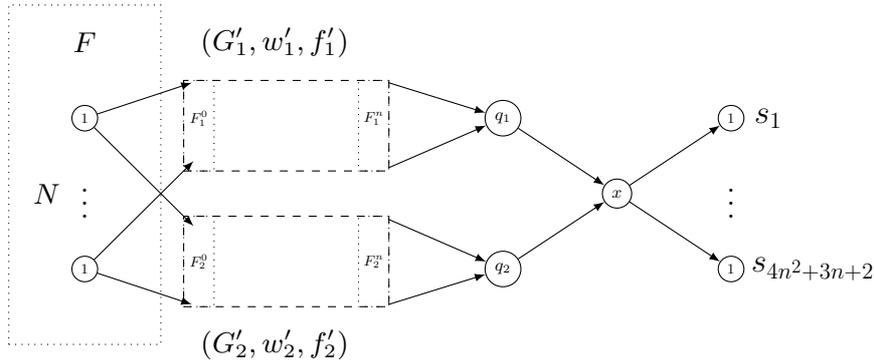

Now, given two influence games $\Gamma_1=(G_1,w_1,f_1,q_1,N)$ and  $\Gamma_2=(G_2,w_2,f_2,q_2,N)$ we construct the  two influence graphs $(G_1',w_1',f_1')$ and $(G_2',w_2',f_2')$ as described before (see Figure~\ref{fig:SG-Fi(X)}).
We use the construction depicted in   Figure~\ref{fig:UI-IGs} to construct  another  influence graph. In this last construction we add a set $F$ which is a copy of $N$. All the  nodes  in $F$ have  label 1. The nodes in  $F$ are connected to their corresponding nodes in  $F_1^0$ and  in $F^0_2$ through edges with weight 1. Furthermore, we add  a node with label $q_1$, a node with label $q_2$, a node with label $x$,  and a set with  $4n^2+3n+2$ nodes. Those new nodes are connected according to the pattern given in  Figure~\ref{fig:UI-IGs}. 
The nodes in the last column, $F^n_i$,  of $(G_i',w_i',f_i')$ are all connected to the node with label $q_i$, for $i\in\{1,2\}$.  
The nodes with labels $q_1$  and $q_2$ are connected to the node with label $x$ which is connected to the last set of nodes.
All those new connections have assigned weight 1.
Observe that in total we have at most 
$2(2n^2+n)+n+3+4n^2+3n+2$ nodes. Thus the overall  construction can be computed in polynomial time.

Let $(G_\cup,w_\cup,f_\cup)$ be the influence graph obtained by setting $x=1$ and  $(G_\cap,w_\cap,f_\cap)$ be the influence graph  obtained by setting $x=2$. Consider the games $\Gamma_\cup=(G_\cup,w_\cup,f_\cup, 4n^2+3n+2,F)$ and  $\Gamma_\cap=(G_\cap,w_\cap,f_\cap, 4n^2+3n+2,F)$.
By construction a team $X$ is successful in $\Gamma_\cup$  if and only if $X$ is succesful in either $\Gamma_1$ or $\Gamma_2$. Further, a team is successful in $\Gamma_\cap$ if and only if $X$ is succesful in both $\Gamma_1$ and $\Gamma_2$.
\end{proof}

It is interesting to note that it is possible to devise a construction representing the  intersection or the union of weighted games as the influence games $(G'_\cup,w',f'_\cup,n+2,N)$ and $(G'_\cap,w',f'_\cap,n+3,N)$. The corresponding influence graphs $(G'_\cup,w',f'_\cup)$  and $(G'_\cap,w',f'_\cap)$ are shown in Figure~\ref{fig:Int-Uni} (setting as before label $x$ to be 1 or 2 depending on the considered operation). This new  construction requires only a linear number of additional nodes, however the graph is weighted.

\begin{figure}[t]
\centering
\begin{tikzpicture}[every node/.style={circle,scale=0.6}, >=latex]
\node[draw](a)  at (0,4)  {1};
\node[scale=2]  at (1.5,4){$\ldots$};
\node[draw](b)  at (3,4)  {1};
\node[scale=2]  at (4.5,4){$\ldots$};
\node[draw](c)  at (6,4)  {1};
\node[draw,scale=1.2](l) at (1.5,2) {$q_1$};
\node[draw,scale=1.2](m) at (4.5,2) {$q_2$};
\node           at (0.4,3.1){\Large{$w^1_1$}};
\node           at (1.2,3.2){\Large{$w^2_1$}};
\node           at (2.4,3.5){\Large{$w^1_i$}};
\node           at (3.6,3.5){\Large{$w^2_i$}};
\node           at (4.6,3.2){\Large{$w^1_n$}};
\node           at (5.6,3.1){\Large{$w^2_n$}};
\node[draw,scale=1.5](r) at (3,1) {$x$};
\node[draw](aa) at (1.5,0){1};
\node[scale=2]  at (3.0,0){$\ldots$};
\node[draw](cc) at (4.5,0){1};
\draw[->] (a) to node {}(l);
\draw[->] (b) to node {}(l);
\draw[->] (c) to node {}(l);
\draw[->] (a) to node {}(m);
\draw[->] (b) to node {}(m);
\draw[->] (c) to node {}(m);
\draw[->] (l) to node[label=right:{\Large $1$}] {}(r);
\draw[->] (m) to node[label=left: {\Large $1$}] {}(r);
\draw[->] (r) to node[label=right:{\Large $1$}] {}(aa);
\draw[->] (r) to node[label=left: {\Large $1$}] {}(cc);
\node at (8.5,4){\Large{$n$ nodes}};
\node at (8.5,0){\Large{$n$ nodes}};
\end{tikzpicture}
\caption{Influence graphs associated  to
$[q_1;w^1_1,\ldots,w^1_n]$ $\cap$ $[q_2;w^2_1,\ldots,w^2_n]$ ($x=2$), and $[q_1;w^1_1,\ldots,w^1_n]$ $\cup$ $[q_2;w^2_1,\ldots,w^2_n]$ ($x=1$).\label{fig:Int-Uni}}
\end{figure}
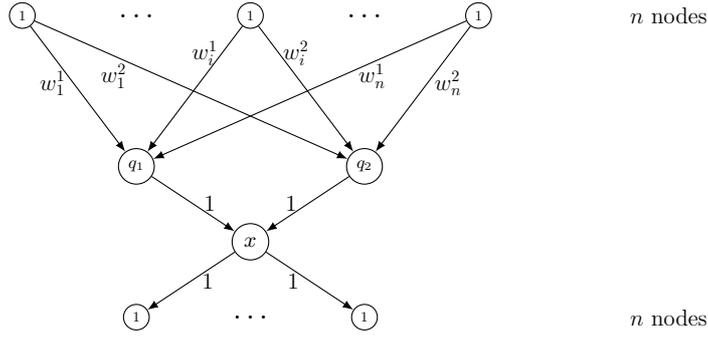

Thus, as any simple game can be represented as the intersection or union of a finite number of weighted games, we have an alternative way to show the completeness of the family of weighted influence games with respect to the class of simple games (Theorem~\ref{the:SG-Inf}).
However, as the dimension, the codimension, and the representation as boolean weighted game of a simple game might be exponential in the number of players (but bounded by the number of maximal losing, minimal winning coalitions, or both,  respectively)~\cite{FP01,FM09a,FEW09}, we cannot conclude that any simple game can be represented by a weighted  influence game whose number of agents is polynomial in the number of players. For the particular case of unweighted influence game we know the following.

\begin{theorem}\label{teo:polyIG}
The family of unweighted influence games in which the number of agents in the corresponding influence graph is polynomial in the number of players is a proper subset of simple games.
\end{theorem}
\begin{proof}
We use a  simple counting argument to show the result. Observe that, for any $n\geq0$, there are  more than $2^{(2^n/n)}$ simple games with $n$ players~\cite{KM75}.  Taking into account that a simple game has at most $n!$ isomorphic simple games we know that there are more that  $2^{(2^n/n)}/n!$ different simple games on $n$ players.

Consider an unweighted  influence game with $n$ players and $f(n)$ agents.  The possibilities for the edge sets are less than $2^{(f(n)+1)^2}$.
 It  suffices to consider label functions assigning values between 0 and $f(n)+1$.   Thus, there are at most $(f(n)+2)^{f(n)+2}$
possibilities for the labeling functions. Finally,  for the quota, only $f(n)+2$ possibilities have to be considered.
Thus, the number of unweighted  influence games with $n$ players and $f(n)$ agents is at most $2^{O(f(n)^2)}$.

Taking $f(n)= n^{\log n}$, the family includes all unweighted  influence games with $n$ players and polynomial number of agents.  Taking the logarithm on both sides, one easily sees that  $2^{O(f(n)^2)}$ is asymptotically smaller than  $2^{(2^n/n)}/n!$. 
\end{proof}

\section{Parameters and Properties}\label{sec:par-pro}

From Theorems~\ref{the:SG-Inf} and \ref{the:WG-Inf} we know that all the computational problems related to properties and parameters that are computationally hard for simple games in winning or minimal winning form, as well as for weighted games, are also computationally hard for influence games. Nevertheless, the hardness results do not apply to unweighted influence games with polynomial, in the number of players,  number of agents. In this section we address the computational complexity of problems for games with a polynomial number of agents. 
All the hardness proofs are given for the  subclass formed by unweighted influence games on undirected influence graphs, which is a subset of all the other variations. The polynomial time algorithms are devised for the biggest class of general influence games, i.e., weighted influence games on directed graphs which includes all others.

Before starting to analyze problems we state here some basic results. From Lemma~\ref{lem:Fpoly} we know that, for a given team $X$, we can compute in polynomial time the set $F(X)$. Therefore we have the following. 

\begin{lemma}
For a given influence game $(G,w,f,q,N)$, deciding whether a team $X\subseteq N$ is successful can be done in polynomial time. 
\end{lemma}  

Our next result concerns a particular type of influence games that we will use first as a basic construction,   which associates an unweighted influence game  to an undirected graph, and later as a representative of a particular subclass of influence games.
\begin{definition}
Given an undirected graph $G=(V,E)$, the unweighted influence game $\Gamma(G)$ is the game $(G,f,|V|,V)$ where, for any $v\in V$, the label $f(v)$ is the degree of $v$ in $G$, i.e., $f(v)=d_G(v)$.
\end{definition}
Recall that a set $S\subseteq V$ is a \emph{vertex cover} of a graph $G$ if and only if,  for any  edge  $(u,v)\in E$, $u$ or $v$  (or both) belong to $S$. 
From the definitions we get the following result.
\begin{lemma}\label{lem:gammaG}
Let $G$ be an undirected graph. A team $X$ is successful in $\Gamma(G)$ if and only if $X$ is a vertex cover of $G$, Furthermore,  the influence game $\Gamma(G)$ can be obtained in polynomial time, given a description of $G$.
\end{lemma}

Table~\ref{tab:results} in Section \ref{sec:intro} summarizes the known computational complexity results (notation from~\cite{GJ79})
over parameters and properties~\cite{DP94,FK96,MM00,Pol08,Azi09,FMOS12,RP11}.
For the hardness results we provide polynomial time reductions from the following problems (decision or optimization versions) which are known to be  NP-hard~\cite{GJ79}.

\begin{tabbing}
 xxxxxxxxxxxxxxxx    \= xxxxxxxxxxxxxxxxxxxxxxxxxxxxxxxxxxxxxxxxxxxxxxxx \kill
 {\sc Vertex Cover}: \> Given an undirected graph $G=(V,E)$ and an integer $k$. \\
                     \> Does $G$ have a vertex cover with size $k$ or less?
\end{tabbing}

\begin{tabbing}
 xxxxxxxxxxxxxxxx    \= xxxxxxxxxxxxxxxxxxxxxxxxxxxxxxxxxxxxxxxxxxxxxxxx \kill
 {\sc Set Cover}:    \> Given a finite set $S$, a collection of subsets $C\seb S$, and an integer $k$. \\
                     \> Is there a subset $C'\seb C$ with $|C'|\leq k$ such that\\
                     \> every element in $S$ belongs to at least one member of $C'$?
\end{tabbing}

\begin{tabbing}
 xxxxxxxxxxxxxxxx    \= xxxxxxxxxxxxxxxxxxxxxxxxxxxxxxxxxxxxxxxxxxxxxxxx \kill
 {\sc Set Packing}:  \> Given a collection $C$ of finite sets, and an integer $k$. \\
                     \> Is there a collection of disjoint sets $C'\seb C$ with $|C'|\geq k$?
\end{tabbing}

We first analyze the complexity of two relevant simple game measures that provide information about the size of the teams to succeed or not.
Further, we consider  two new measures which are relevant for influence games: the {\em strict length}, defined in Section~\ref{sec:IG}, and its dual, the {\em strict width}.
Assume that a simple game $\Gamma$ is given.

\begin{tabbing}
 xxxxxxxxxxxxxxxx    \= xxxxxxxxxxxxxxxxxxxxxxxxxxxxxxxxxxxxxxxxxxxxxxxx \kill
 {\sc Length}:       \> The minimum size of a successful team, i.e.,  $\min\,\{|X| \mid X\in \cW \}$.
\end{tabbing}

\begin{tabbing}
 xxxxxxxxxxxxxxxx    \= xxxxxxxxxxxxxxxxxxxxxxxxxxxxxxxxxxxxxxxxxxxxxxxx \kill
 {\sc Width}:        \> The maximum size of an unsuccessful team, i.e.,  $\max\,\{|X| \mid X\in \cL \}$.
\end{tabbing}

\begin{tabbing}
 xxxxxxxxxxxxxxxx    \= xxxxxxxxxxxxxxxxxxxxxxxxxxxxxxxxxxxxxxxxxxxxxxxx \kill
 {\sc sLength}:      \> The minimum cardinality from which all teams are successful,\\
		     \> i.e., $\min\,\{k\in \mathbb{N} \mid \cP_k(N) \subseteq \cW \}$.
\end{tabbing}

\begin{tabbing}
 xxxxxxxxxxxxxxxx    \= xxxxxxxxxxxxxxxxxxxxxxxxxxxxxxxxxxxxxxxxxxxxxxxx \kill
 {\sc sWidth}:       \> The maximum cardinality to which all teams are unseccessful,\\
		     \> i.e., $\max\,\{k\in\mathbb{N} \mid \cP_k(N) \subseteq \cL \}$.
\end{tabbing}

Observe that for every simple game $\Gamma$, we have that $\textsc{Width}(\Gamma)=\mbox{\sc sLength}(\Gamma)-1$ and $\mbox{\sc Length}(\Gamma)=\mbox{\sc sWidth}(\Gamma)+1$. As we have shown in Theorem~\ref{lem:slength-min-win},  $\textsc{ sLength}(\Gamma)$ can be computed in polynomial time when the game is given by $(N,\cW)$ or $(N,\cW^m)$.  Our next result shows that computing all the above measures is hard for influence games.

\begin{theorem}\label{the:length-width-hard}
Computing {\sc Length}, {\sc Width}, {\sc sLength} and {\sc sWidth} of an unweighted influence game is {\sc NP}-hard.
\end{theorem}

\begin{proof}
For the measure {\sc Length}, we provide a reduction from the minimum set cover problem. Let $C=\{C_1,\ldots,C_m\}$ be a collection of subsets of a universe with $n$ elements.
We associate to $C$ the unweighted influence game $(G,f,q,N)$ where $G=(V,E)$. The graph $G$ has three disjoint sets of vertices:  $Y=\{y_1,\dots, y_m\}$,  $T=\{t_1,\dots, t_n\}$,   and $Z=\{z_1,\dots, z_{m+1}\}$, together with an additional vertex $x$. The components of the game are the following.
\begin{itemize}
\item $V=Y \cup T\cup \{x\}  \cup Z $,
\item 
$E=\{(y_j,t_i)\mid i\in C_j\} \cup\{(t_i,x)\mid 1\leq i\leq n\} \cup\{(x,z_k)\mid 1\leq k\leq m+1\}$,
\item $f(y_j)=n+1$, for any $1\leq j\leq m$,
\item $f(t_i)=1$, for any $1\leq i\leq n$,
\item $f(z_k)=1$, for any $1\leq k\leq m+1$,
\item $f(x)=n$,
\item $q=m+n+1$ and
\item $N=Y$.
\end{itemize}
Therefore, it is easy to see that a team $X\seb N$ succeeds if and only if it corresponds to a set cover, so the \mbox{\sc Length} of $(G,f,q,N)$ coincides with the size of a minimum set cover.

For the measure {\sc Width} we provide a reduction from the maximum set packing problem. Consider an influence game $(G',f',q',N)$  where $G'$ is constructed from $G$. We remove node $\{x\}$, add the connections $\{(t_i,z_k)\mid 1\leq i\leq n, 1 \leq k\leq  m+1\}$, and  set $f'(t_i)=2$ for any $1\leq i\leq n$. We keep $N=Y$ and set $q'=m+1$.
It is easy to see that a team $X\seb N$ is unsuccessful in $(G',f',q',N)$ if and only if  $X$  corresponds to a set packing in $C$. Therefore, the {\sc Width} of  $(G',f',q',N)$ is  the size of a maximum set packing of $C$.

The remaining results for {\sc sLength} and {\sc sWidth} follow from $\mbox{\sc Width}(\Gamma)=\mbox{\sc sLength}(\Gamma)-1$ and $\mbox{\sc Length}(\Gamma)=\mbox{\sc sWidth}(\Gamma)+1$.
\end{proof}

The hardness result for {\sc Length} can be obtained directly from Lemma~\ref{lem:gammaG}. There we provide a reduction from the {minimum vertex cover problem}. However, the reduction from  the {minimum set cover}  problem given in the previous theorem allows us to extract additional results about the complexity of approximation.   In particular, the reductions in  Theorem~\ref{the:length-width-hard} imply that {\sc Length} is neither approximable within $(1-\epsilon)\cdot\log m$ nor within $c\cdot\log n$, for some $c>0$, and that {\sc Width} is not approximable within $m^{1/2-\epsilon}$, for any $\epsilon>0$,  using the non-approximability  results from \cite{ACGKMP99} for the problems minimum set cover and minimum set packing.

When the simple game is given  by $(N,\cW)$ or $(N,\cW^m)$, {\sc Length} can trivially be computed in polynomial time. As $\mbox{\sc Length}(\Gamma)=\mbox{\sc sWidth}(\Gamma)+1$, {\sc sWidth} can also be computed in polynomial time. 
Computing {\sc Width} for simple games given by $(N,\cW)$ or $(N,\cW^m)$ was posted as open problem in \cite{Azi09}.
However,  from  Theorem~\ref{lem:slength-min-win} and taking into account that  $\mbox{\sc Width}(\Gamma)=\mbox{\sc sLength}(\Gamma)-1$ we have that {\sc Width} can be computed in polynomial time.  

\begin{theorem}\label{the:width-hard}
Given a simple game by $(N,\cW)$ or $(N,\cW^m)$,  {\sc Length}, {\sc Width}, {\sc sLength} and {\sc sWidth} can be computed in polynomial time.
\end{theorem}

Our next result settles the complexity of the computation of the Banzhaf and Shapley-Shubik values of a given player.

\begin{definition}[\cite{Fre11}]
Let $\Gamma=(N,\cW)$ be a simple (influence)  game. For any  $i\in N$,  let $C_i=\{S\in\cW; S\bac\{i\}\notin\cW\}$ be the set of blocking teams (coalitions) for player $i$.
The {\em Banzhaf value} of $i$ on $\Gamma$ is $\eta_i(\Gamma)=|C_i|$ and the
{\em Shapley-Shubik value} of $i$ on $\Gamma$ is $$\kappa_i(\Gamma)=\sum_{S\in C_i}(|S|-1)!\,(n-|S|)!$$
The {\em Banzhaf index} of $i$ on $\Gamma$ is $\eta_i(\Gamma)/2^{n-1}$ and the
{\em Shapley-Shubik index} of $i$ on $\Gamma$ is $\kappa_i(\Gamma)/{n!}$.
\end{definition}

The computational complexity of the two indices is the same as that of the corresponding value.  We analyze here only the problems of computing the values. The two problems are denoted as {\sc Bval} and {\sc SSval}, respectively.

\begin{theorem}\label{the:Ban-SS}\label{the4}
Computing {\sc Bval} and {\sc SSval} for a given influence game and a given player is \#{\sc P}-complete.
\end{theorem}

\begin{proof} Both problems belong trivially to \#{\sc P}. To show hardness we construct a reduction for the problem of computing the number of vertex covers of a given graph which is known to be \#P-complete~\cite{GJ79}. 
Let $G$ be a graph, we first construct the graph $G'$ which is obtained from $G$ adding a  new vertex $x$ and connecting $x$ to all the vertices in $G$. The associated input to {\sc Bval} is formed by the influence game $\Gamma(G')$ and the player $x$. Observe that the reduction can be computed in polynomial time. 

Let $X$ be a successful team in  $\Gamma(G')$ such that $x\in X$. When $X\neq V(G')$ we know that $X\setminus \{x\}$  must be a vertex cover of $G$. Furthermore $x\in C_x$  as $X\setminus \{x\}$ is not winning in $\Gamma(G')$. When $X= V(G')$,  $X\setminus \{x\}$  is winning in $\Gamma(G')$ and thus  $x\notin C_x$.
As a consequence, we have that $\eta_x(\Gamma)$ coincides with the number of vertex covers of $G$ minus one. As computing the number of vertex covers of a graph is \#P-hard, we have that {\sc Bval} is \#P-hard.

According to \cite{Azi09} (Theorem 3.29, page 50), to prove that {\sc SSval} is \#P-hard, it is enough to show that {\sc Bval} is \#P-hard and that influence games verify the property of being a {\em reasonable representation}. This last condition is stated as follows.  For a simple game $\Gamma= (N,\cW)$, consider the game $\Gamma'=(N\cup\{x\},\cW')$, where $x$ is a new player and   $\cW'=\{S\cup\{x\}\mid S\in \cW\}$.  A representation of simple games is \emph{reasonable} if a representation of the  game $\Gamma'$ can be computed with only polynomial blow-up with respect to a given representation of the game $\Gamma$. In the  remaining of this proof we show that influence games are a reasonable representation.

Let $\Gamma=(G,w,f,q,N)$ be an influence game, and assume that $G=(V,E)$ has $n$ vertices and $m$ edges.
Consider the influence graph $(G',w',f')$ where 
\begin{itemize}
\item $G'=(V',E')$ and $V'= V \cup \{x,y\} \cup \{a_1,\dots , a_{2n}\}$,
\item $E'= E\cup \{(x,y)\} \cup \{(v,y)\mid v\in V\}\cup \{(y,a_i)\mid 1\leq i\leq 2n\}$,
\item $w'(e)= w(e)$, for any $e\in E$, and $w'(e)= 1$, for any $e\in E'\setminus E$, and
\item $f'(v)=f(v)$, for any $v\in V$, $f'(x)=1$, $f'(y)=q+1$, and $f'(a_i)=1$, for any $1\leq i\leq 2n$.
\end{itemize}
Finally, we consider the influence  game $\Gamma^+=(G',w',f',q',N')$ where $q'=2n$ and $N'=N\cup\{x\}$.

From the previous construction, it follows  that all the winning coalitions in  $\Gamma^+$ must include $x$. Furthermore, $X\cup \{x\}$ is a winning coalition in $\Gamma^+$ if and only if  $X$ is a winning coalition in $\Gamma$. Therefore, $\Gamma^+$ is a representation of $\Gamma'$  and has polynomial size with respect to the size of $\Gamma$. So,  we conclude that  influence games are a  reasonable representation.
\end{proof}

Now we consider another set of problems that reflect fundamental properties of simple games. The input to the following problems is a game $\Gamma$.

\begin{tabbing}
 xxxxxxxxxxxxxxxx    \= xxxxxxxxxxxxxxxxxxxxxxxxxxxxxxxxxxxxxxxxxxxxxxxx \kill
 {\sc IsProper}:     \> Determine whether $\Gamma$ is proper, i.e., whether, for any $X\in \cW$, $N\setminus X\in \cL$. 
\end{tabbing}

\begin{tabbing}
 xxxxxxxxxxxxxxxx    \= xxxxxxxxxxxxxxxxxxxxxxxxxxxxxxxxxxxxxxxxxxxxxxxx \kill
 {\sc IsStrong}:     \> Determine whether $\Gamma$ is strong, i.e., whether, for any $X\in \cL$, $N\setminus X\in \cW$. 
\end{tabbing}

\begin{tabbing}
 xxxxxxxxxxxxxxxx    \= xxxxxxxxxxxxxxxxxxxxxxxxxxxxxxxxxxxxxxxxxxxxxxxx \kill
 {\sc IsDecisive}:   \> Determine whether $\Gamma$ is decisive, i.e., whether $X\in \cW$ if and only if $N\setminus X\in \cL$. 
\end{tabbing}

When a simple game $\Gamma$ is given by $(N,\cW^m)$, it is known that the \textsc{IsProper} problem can be decided in polynomial time. Here we can check,   for any $X\in\cW^m$, $N\bac X\notin\cW$ in polynomial time. Further, the \textsc {IsStrong} problem is  coNP-complete~\cite{Pol08} and the \textsc{IsDecisive} problem  can be solved in  quasi-polynomial time~\cite{FK96} but is not known to be polynomial time solvable. All these results were obtained in the context of simple game theory~\cite{RP11,FMOS12}.

The following problems consider properties of a player with respect to an influence game. Their input is composed of  a simple game $\Gamma$ and player $i$.

\begin{tabbing}
 xxxxxxxxxxxxxxxx    \= xxxxxxxxxxxxxxxxxxxxxxxxxxxxxxxxxxxxxxxxxxxxxxxx \kill
 {\sc IsDummy}:      \>  Determine whether player $i$ is a dummy player,\\
\>  i.e., whether, for any team $X\in \cW$ with $i\in X$,  $X\setminus\{i\}\in \cW$.
\end{tabbing}

\begin{tabbing}
 xxxxxxxxxxxxxxxx    \= xxxxxxxxxxxxxxxxxxxxxxxxxxxxxxxxxxxxxxxxxxxxxxxx \kill
 {\sc IsPasser}:     \>  Determine whether player $i$ is a passer,\\
\>  i.e., whether, for any team $X$ with $i\in X$,  $X\in \cW$.
\end{tabbing}

\begin{tabbing}
 xxxxxxxxxxxxxxxx    \= xxxxxxxxxxxxxxxxxxxxxxxxxxxxxxxxxxxxxxxxxxxxxxxx \kill
 {\sc IsVetoer}:     \>   Determine whether player $i$ is  a veto player, \\
\> i.e., whether, for any team $X$ with $i\notin X$,  $X\in \cL$.
\end{tabbing}

\begin{tabbing}
 xxxxxxxxxxxxxxxx    \= xxxxxxxxxxxxxxxxxxxxxxxxxxxxxxxxxxxxxxxxxxxxxxxx \kill
 {\sc IsDictator}:   \>   Determine whether player $i$ is a dictator,\\
\> i.e., whether, $X\in \cW$ if and only if $i\in X$. 
\end{tabbing}

\begin{tabbing}
 xxxxxxxxxxxxxxxx    \= xxxxxxxxxxxxxxxxxxxxxxxxxxxxxxxxxxxxxxxxxxxxxxxx \kill
 {\sc IsCritical}:   \> Given, in addition,  a team $X$ including $i$,  determine whether $i$ is  critical for $X$,\\
\>  i.e., $X\in \cW$  and $X\setminus\{i\}\in \cL$.
\end{tabbing}

\begin{tabbing}
 xxxxxxxxxxxxxxxx    \= xxxxxxxxxxxxxxxxxxxxxxxxxxxxxxxxxxxxxxxxxxxxxxxx \kill
 {\sc AreSymmetric}: \> Given, in addition,  an agent $j$, determine whether players $i$ and $j$ are symmetric, \\
\>i.e., whether,  for any team $X$, $X\cup\{i\}\in \cW$  if and only if $X\cup\{j\}\in \cW$.
\end{tabbing}

The following problems consider properties of a team of agents  $X$ with respect to an influence game. Now the input is composed of a game $\Gamma$ and a team $X$.  

\begin{tabbing}
 xxxxxxxxxxxxxxxx    \= xxxxxxxxxxxxxxxxxxxxxxxxxxxxxxxxxxxxxxxxxxxxxxxx \kill
 {\sc IsBlocking}:   \> Determine whether $X$ is a blocking team, i.e., whether $N\setminus X\in \cL$.
\end{tabbing}

\begin{tabbing}
 xxxxxxxxxxxxxxxx    \= xxxxxxxxxxxxxxxxxxxxxxxxxxxxxxxxxxxxxxxxxxxxxxxx \kill
 {\sc IsSwing}:      \> Determine whether $X$ is a swing, i.e., whether $X\in \cW$ but there is $i\in X$ such that\\
		     \> $X\setminus \{i\}\in\cL$. 
\end{tabbing}

To analyze the complexity of the above problems we consider a new construction.
Let $G=(V,E)$ be a graph where $V=\{v_1,\dots,v_n\}$ and $E=\{e_1,\dots, e_m\}$, and let $k$ be an integer (which will be useful to consider a set cover of size $k$ or less).
Then the unweighted influence game $\Delta_1(G,k)=(G_1,f_1,q_1,N_1)$ is defined as follows, where Figure~\ref{fig:Delta} shows the corresponding influence graph.

\begin{figure}[t]
\centering
\begin{tikzpicture}[every node/.style={circle,scale=0.7}, >=latex]

\node[draw](v1)  at (1,3.5)[label=left:\Large{$v_1$}] {$m+2$};
\node[draw](vn)  at (1,2.5)[label=left:\Large{$v_{n}$}] {$m+2$};

\node[draw](x)   at (3.5,5.5)[label=above:\Large{$x$}] {$k+1$};
\node[draw](e1)  at (3,4.3)[label=right:\Large{$e_1$}] {$1$};
\node[draw](em)  at (3,1.7)[label=right:\Large{$e_{m}$}] {$1$};
\node[draw](z)   at (4,1)[label=right:\Large{$z$}] {$2$};

\node[draw](y)   at (5,3)[label=left:\Large{$y$}] {$m+1$};

\node[draw](s1)  at (8,4)[label=right:\Large{$s_1$}] {$1$};
\node[scale=1.5] at (8,3.5){$\ \ \,\vdots\ \ $};
\node[scale=1.5] at (8,3.0){$\ \ \,\vdots\ \ $};
\node[scale=1.5] at (8,2.6){$\ \ \,\vdots\ \ $};
\node[draw](smn4)at (8,2)[label=right:\Large{$s_{n+m+4}$}] {$1$};

\draw[-](v1) to node{}(x);
\draw[-](vn) to node{}(x);

\node[scale=0.8](ig)   at (2,3){ $G$'s incidence graph};
\draw[-,dotted](v1) to node{}(e1);
\draw[-,dotted](v1) to node{}(vn);

\draw[-,dotted](e1) to node{}(em);
\draw[-,dotted](vn) to node{}(em);

\draw[-](x) to node{}(s1);
\draw[-](x) to node{}(smn4);

\draw[-](e1) to node{}(y);
\draw[-](em) to node{}(y);
\draw[-](z) to node{}(y);

\draw[-](y) to node{}(s1);
\draw[-](y) to node{}(smn4);
\end{tikzpicture}
\caption{Influence graph $(G_1,f_1)$ of the game $\Delta_1(G,k)$.\label{fig:Delta}}
\end{figure}
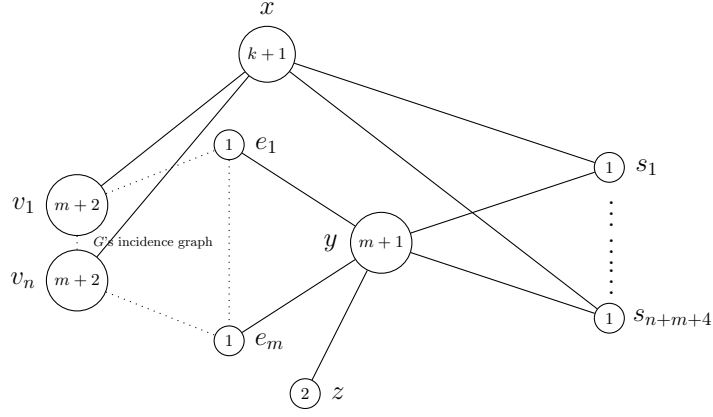

The graph $G_1=(V_1,E_1)$ has 
$V_1=\{v_1,\dots,v_n, e_1,\dots, e_m, x,y,z, s_1,\dots,s_\alpha\}$
where $\alpha=m+n+4$. The edges in $E_1$ are constructed as follows.
We include the incidence graph of $G$: for any $e=(v_i,v_j)\in E$, we add to $E_1$ the edges $(e,v_i)$, $(e,v_j)$ and $(e,y)$.
For any $1\leq i\leq n$, we add the edge $(v_i,x)$.
For any $1\leq j\leq \alpha$, we add  the edges $(x,s_j)$ and $(y,s_j)$.
Finally, we add the edge $(z,y)$.
The labeling function $f_1$ is  defined as: $f_1(v_i)=m+2$, $1\leq i\leq n$;
$f_1(e_j)=1$, $1\leq j \leq m$;
$f_1(s_\ell)=1$, $1\leq \ell\leq \alpha$; and 
$f_1(z)=2$, $f_1(x)=k+1$, $f_1(y)=m+1$.
The quota  is  $q_1=\alpha$ and the set of players is $N_1=\{v_1,\dots,v_n, z\}$.

Observe that by construction the games $\Gamma(G)$ and $\Delta_1(G,k)$ can be obtained in polynomial time.
As an immediate consequence of the definition, we have that
$X$ is a successful team in $\Delta_1(G,k)$ if and only if either $(|X\cap V|\ge k+1)$ or $z\in X$ and $X\setminus z$ is a vertex cover  in $G$.

Our next result settles the complexity of the problems which are {\sc coNP}-complete.

\begin{theorem}\label{lem:dummy}
For unweighted influence games with polynomial number of vertices, the problems {\sc AreSymmetric}, {\sc IsDummy}, {\sc IsProper}, {\sc IsStrong} and {\sc IsDecisive} are  {\sc coNP}-complete.
\end{theorem}

\begin{proof}
Membership in {\sc coNP} follows from the definitions.
To get the hardness results, we provide reductions from the complement of the {\sc Vertex Cover} problem and some other problems derived from it. Let $(G,k)$ be an input to {\sc Vertex Cover}, as usual we assume that $G$ has $n$ vertices and $m$ edges. 

Let us start considering the  {\sc IsDummy} problem. Starting from $G=(V,E)$ and  $k$, we construct the unweighted influence game
$\Delta_1(G,k)$ and the  pair $(\Delta_1(G,k), z)$ which is an instance of the {\sc IsDummy} problem. 
If $G$ has a vertex cover $X$ with size $k$ or less, by construction, we have that $X\cup\{z\}$ is a successful team of $\Delta_1(G,k)$. Furthermore, if $X$ is a vertex cover of minimum size, we have that $X\cup\{z\}$ is a minimal successful team. Therefore, $z$ is not a dummy player in $\Delta_1(G,k)$. If $G$ does not have a vertex cover with size $k$ or less and $X$ is a successful team containing $z$, it must hold that $|X\setminus\{z\}| > k$, therefore $X\setminus\{z\}$ is a successful team. In consequence $z$ is a dummy player in $\Delta_1(G,k)$.
As  the pair $(\Delta_1(G,k), z)$ is computable in polynomial time, we have the desired result.

Let us consider now the  {\sc AreSymmetric} problem. Starting from $G=(V,E)$ and  $k$, we construct the unweighted influence game $\Delta_2(G,k)=(G_2,f_2,q_2,N_2)$  (see Figure~\ref{fig:Delta2}).
$G_2$ is obtained from the graph $G_1$ appearing in the construction of $\Delta_1(G,k)$
by adding two new vertices $t$ and $s$ and the edges $(x,s)$, $(y,s)$ and $(t,s)$. Recall that $V(G_1)=\{v_1,\dots,v_n, e_1,\dots, e_m, x,y,z, s_1,\dots,s_\alpha\}$.The label function is the following: $f_2(v)=f_1(v)$, for $v\in V(G_2)\cap V(G_1)$; $f_2(s)=4$; $f_2(t)=2$. Finally, $q_2=\alpha+1 = n+m+5$ and $N_2=\{v_1,\dots,v_n,z,t\}$. Note that  a description of $G_2$ can be obtained in polynomial time as well as a description of   $\Delta_2(G,k)$ given a description of $(G,k)$. Let us show that the construction is indeed a reduction.

When $G$ has a vertex cover $X$ of size $k$ or less, by construction the team $X\cup \{z\}$ is successful in $\Delta_2(G,k)$ while the team $X\cup \{t\}$ is unsuccessful. Therefore $z$ and $t$ are not symmetric.
When $G$ does not have a vertex cover $X$ of size $k$ or less, by construction, any successful team $Y$ must contain a subset with at least $k+1$ vertices from $\{v_1,\dots, v_n\}$. Therefore both $Y\cup \{z\}$ and $Y\cup \{t\}$ are successful teams in $\Delta_2(G,k)$, i.e., vertices $z$ and $t$ are  symmetric.

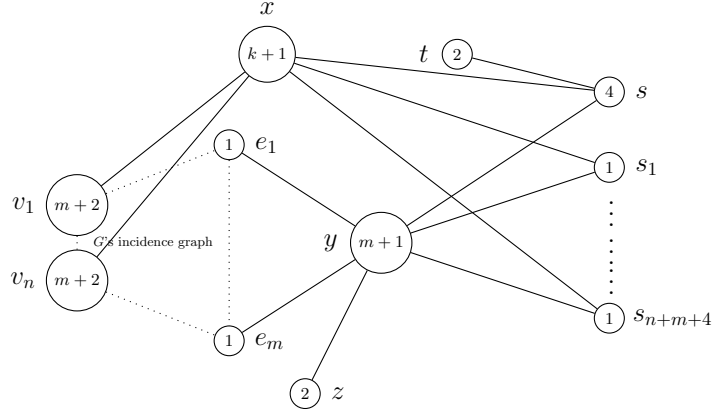
\begin{figure}[t]
\centering
\begin{tikzpicture}[every node/.style={circle,scale=0.7}, >=latex]

\node[draw](v1)  at (1,3.5)[label=left:\Large{$v_1$}] {$m+2$};
\node[draw](vn)  at (1,2.5)[label=left:\Large{$v_{n}$}] {$m+2$};

\node[draw](x)   at (3.5,5.5)[label=above:\Large{$x$}] {$k+1$};
\node[draw](e1)  at (3,4.3)[label=right:\Large{$e_1$}] {$1$};
\node[draw](em)  at (3,1.7)[label=right:\Large{$e_{m}$}] {$1$};
\node[draw](z)   at (4,1)[label=right:\Large{$z$}] {$2$};

\node[draw](y)   at (5,3)[label=left:\Large{$y$}] {$m+1$};
\node[draw](t)   at (6,5.5)[label=left:\Large{$t$}] {$2$};

\node[draw](s)   at (8,5)[label=right:\Large{$s$}] {$4$};
\node[draw](s1)  at (8,4)[label=right:\Large{$s_1$}] {$1$};
\node[scale=1.5] at (8,3.5){$\ \ \,\vdots\ \ $};
\node[scale=1.5] at (8,3.0){$\ \ \,\vdots\ \ $};
\node[scale=1.5] at (8,2.6){$\ \ \,\vdots\ \ $};
\node[draw](smn4)at (8,2)[label=right:\Large{$s_{n+m+4}$}] {$1$};

\draw[-](v1) to node{}(x);
\draw[-](vn) to node{}(x);

\node[scale=0.8](ig)   at (2,3){ $G$'s incidence graph};
\draw[-,dotted](v1) to node{}(e1);
\draw[-,dotted](v1) to node{}(vn);

\draw[-,dotted](e1) to node{}(em);
\draw[-,dotted](vn) to node{}(em);

\draw[-](x) to node{}(s1);
\draw[-](x) to node{}(smn4);

\draw[-](e1) to node{}(y);
\draw[-](em) to node{}(y);
\draw[-](z) to node{}(y);

\draw[-](y) to node{}(s1);
\draw[-](y) to node{}(smn4);

\draw[-](x) to node{}(s);
\draw[-](y) to node{}(s);
\draw[-](t) to node{}(s);
\end{tikzpicture}
\caption{Influence graph $(G_2,f_2)$ used in the definition of the game $\Delta_2(G,k)$.}\label{fig:Delta2}
\end{figure}

To prove hardness for the next two problems, \textsc{IsProper} and \textsc{IsDecisive}, we provide a reduction from the following variation of the {\sc Vertex Cover} problem:

\begin{tabbing}
 xxxxxxxxxxxxxxxxxxxxxx    \= xxxxxxxxxxxxxxxxxxxxxxxxxxxxxxxxxxxxxxxxxxxxxxxx \kill
 {\sc Half vertex cover}:  \> Given an undirected graph with an odd number of vertices $n$.\\
                           \> Is there a vertex cover with size $(n-1)/2$ or less?
\end{tabbing}
We first show that the {\sc Half vertex cover} problem  is NP-complete.
By definition the problem belongs to {NP}. To prove hardness we show a reduction from  the {\sc Vertex Cover} problem.
Given a graph $G$ with $n$ vertices and an integer $k$, $0\leq k\leq n$, we construct a graph $\hat{G}$ as follows (see Figure~\ref{fig:Ghat}).
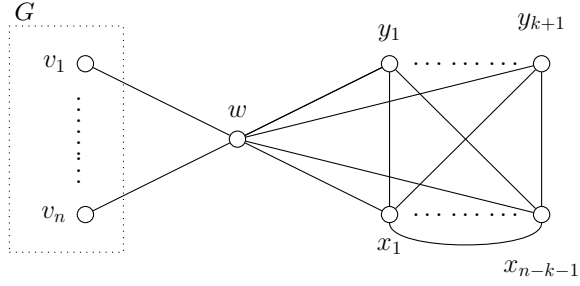
\begin{figure}[t]
\centering
\begin{tikzpicture}[every node/.style={circle,scale=0.7}, >=latex]
\node[scale=1.5] at (0.2,4.7){$G$};
\draw[dotted] (0,1.5) rectangle (1.5,4.5);

\node[draw](v1)  at (1,4   )[label=left:\Large{$v_1$}] {};
\node[scale=1.5] at (1,3.5 ){$\vdots\ \,$};
\node[scale=1.5] at (1,3.06){$\vdots\ \,$};
\node[scale=1.5] at (1,2.7 ){$\vdots\ \,$};
\node[draw](vn)  at (1,2   )[label=left:\Large{$v_{n}$}] {};

\node[draw](w)   at (3,3)[label=above:\Large{$w$}] {};

\node[draw](y1)  at (5  ,4)[label=above:\Large{$y_1$}] {};
\node[scale=1.5] at (5.5,4){$\ldots$};
\node[scale=1.5] at (6.0,4){$\ldots$};
\node[scale=1.5] at (6.5,4){$\ldots$};
\node[draw](yn)  at (7  ,4)[label=above:\Large{$y_{k+1}$}] {};

\node[draw](x1)  at (5  ,2)[label=below:\Large{$x_1$}] {};
\node[scale=1.5] at (5.5,2){$\ldots$};
\node[scale=1.5] at (6.0,2){$\ldots$};
\node[scale=1.5] at (6.5,2){$\ldots$};
\node[draw](xn)  at (7 ,2)[label=below:\Large{$x_{n-k-1}$}] {};

\draw[-](v1) to node{}(w);
\draw[-](vn) to node{}(w);
\draw[-](w) to node{}(x1);
\draw[-](w) to node{}(xn);
\draw[-](w) to node{}(y1);
\draw[-](w) to node{}(y1);
\draw[-](w) to node{}(yn);
\draw[-](x1) to node{}(y1);
\draw[-](x1) to node{}(yn);
\draw[-](xn) to node{}(y1);
\draw[-](xn) to node{}(yn);
\draw[-](x1) .. controls +(down:.5cm) and +(down:.5cm) .. node[above,sloped]{}(xn);

\end{tikzpicture}
\caption{Graph $\hat{G}$ used to prove that {\sc Half vertex cover} is {\sc NP}-hard.}\label{fig:Ghat}
\end{figure}
$\hat{G}$ has vertex set $\hat{V}=V(G) \cup X \cup Y \cup \{w\}$, where $X$ has $n-k-1$ vertices, $Y$ has $k+1$ vertices, and edge set

$$\hat{E}= E \cup \{(x,x') \mid x \neq x'\,\wedge\,x,x'\in X\} \cup \{(x,y)\mid x\in X, y\in Y\} \cup \{(w,z)\mid z\in V \cup X \cup Y \}$$

By construction, $\hat{G}$ has $2n+1$ vertices, so it can be constructed in polynomial time.
Note that any vertex cover $S$ of $\hat{G}$ with minimum size has to contain $w$, all the vertices in $X$ and no vertex from $Y$. The remaining of the cover,  $S\cap V$ must be a minimum vertex cover of $G$. Therefore, $G$ has a vertex cover of size $k$ or less if and only if $\hat{G}$ has a vertex cover of size $n$ or less.

Let us provide a reduction from the \textsc{Half vertex cover} to the \textsc{IsProper} and the \textsc{IsDecisive} problems. 
Let $G$ be an instance of {\sc Half vertex cover} with $2k+1$ vertices, for some value $k\geq 1$. Consider the unweighted influence game $\Delta_1(G,(n-1)/2)=(G_1,f_1,q_1, N_1)$.  Recall that $V(G')=\{v_1,\dots,v_n, e_1,\dots, e_m, x,y,z, s_1,\dots,s_\alpha\}$ where $\alpha=n+m+4$,  $q_1=n+m+5$, and $N_1=\{v_1,\dots,v_n, z\}$. Let $k=(n-1)/2$.

If $G$ has a vertex cover $X$ with $|X|\leq k$, the team $X\cup \{z\}$ is successful and, as  $n+1-|X\cup \{z\}|>k$, we have that $N\setminus (X\cup \{z\})$ is also successful. Hence $\Delta_1(G,k)$ is not proper.
When all the vertex covers of $G$ have more than $k$ vertices, any successful team $Y$ of $\Delta_1(G,k)$
verifies $|Y\cap \{v_1,\dots,v_n\}| > k$, i.e.,  $|Y\cap \{v_1,\dots,v_n\}| \ge k+1$. For a successful team $Y$, we have to consider two cases: $z\in Y$ and $z \notin Y$.
When  $z\in Y$, $N\setminus Y \seb \{v_1,\dots,v_n\}$ and $|N\setminus Y| \leq n-k-1 =k$.  Thus, $N\setminus Y$ is an unsuccessful team.
When $z \notin Y$, $|N\setminus (Y \cup \{z\})|\leq  k$ and $N\setminus Y$ is again an unsuccessful team.
So, we conclude that $\Delta_1(G,(n-1)/2)$ is proper.
As $\Delta_1(G,(n-1)/2)$ can be obtained in polynomial time, the {\sc IsProper} problem is {coNP}-hard.

Observe that when $G$ is an instance of the \textsc{Half vertex cover} and all the vertex covers of $G$ have more than $(n-1)/2$ vertices, the game $\Delta_1(G,(n-1)/2)$ is also decisive. When this condition is not met, the game $\Delta_1(G,(n-1)/2)$ is not proper and thus it is not decisive. Thus, we conclude that the  {\sc IsDecisive}  problem is also coNP-hard. 

\begin{figure}[t]
\centering
\begin{tikzpicture}[every node/.style={circle,scale=0.7}, >=latex]

\node[draw](v1)  at (1,3.5)[label=left:\Large{$v_1$}] {$m+2$};
\node[draw](vn)  at (1,2.5)[label=left:\Large{$v_{n}$}] {$m+2$};

\node[draw](x)   at (3.5,5.5)[label=above:\Large{$x$}] {$k+1$};
\node[draw](e1)  at (3,4.3)[label=right:\Large{$e_1$}] {$2$};
\node[draw](em)  at (3,1.7)[label=right:\Large{$e_{m}$}] {$2$};
\node[draw](z)   at (4,1)[label=right:\Large{$z$}] {$1$};

\node[draw](y)   at (5,3)[label=left:\Large{$y$}] {$1$};
\node[draw](t)   at (6,3)[label=right:\Large{$t$}] {$2$};

\node[draw](s1)  at (8,4)[label=right:\Large{$s_1$}] {$1$};
\node[scale=1.5] at (8,3.5){$\ \ \,\vdots\ \ $};
\node[scale=1.5] at (8,3.0){$\ \ \,\vdots\ \ $};
\node[scale=1.5] at (8,2.6){$\ \ \,\vdots\ \ $};
\node[draw](smn4)at (8,2)[label=right:\Large{$s_{n+m+4}$}] {$1$};

\draw[-](v1) to node{}(x);
\draw[-](vn) to node{}(x);

\node[scale=0.8](ig)   at (2,3){ $G$'s incidence graph};
\draw[-,dotted](v1) to node{}(e1);
\draw[-,dotted](v1) to node{}(vn);

\draw[-,dotted](e1) to node{}(em);
\draw[-,dotted](vn) to node{}(em);

\draw[-](x) to node{}(s1);
\draw[-](x) to node{}(smn4);

\draw[-](e1) to node{}(y);
\draw[-](em) to node{}(y);
\draw[-](z) to node{}(t);
\draw[-](y) to node{}(t);

\draw[-](t) to node{}(s1);
\draw[-](t) to node{}(smn4);
\end{tikzpicture}
\caption{Influence graph $(G_3,f_3)$ of the game $\Delta_3(G,k)$.\label{fig:Delta3}}
\end{figure}
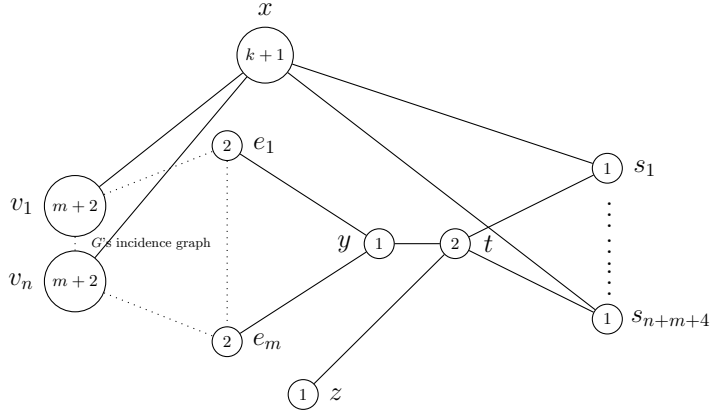

To finish the proof we show hardness for the \textsc{IsStrong} problem. We provide a  reduction  from the complement of the following problem.

\begin{tabbing}
 xxxxxxxxxxxxxxxxxxxxxxxxx    \= xxxxxxxxxxxxxxxxxxxxxxxxxxxxxxxxxxxxxxxxxxxxxxxx \kill
 {\sc Half independent set}:  \> Given an undirected graph with an even number of vertices $n$.\\
                              \> Is there an independent set  with size $n/2$ or higher?
\end{tabbing}

The \textsc{Half independent set} trivially belongs to NP. Hardness  follows from a simple reduction from the \textsc{Half independent set}. Starting from a graph $G$ with an odd number of vertices we construct a new graph $G'$ by adding one new vertex connected to all the vertices in $G$. This construction guarantees that $G$ has a vertex cover of size $(n-1)/2$ or less if and only if $G'$ has a vertex cover with size $n/2$ or less. As the complement of a  vertex cover is an independent set, we have that $G$ has a vertex cover of size $(n-1)/2$  if and only if $G'$ has an independent set with size $n/2$ or higher. 

Now we show that the complement of the \textsc{Half independent set} problem can be reduced to the \textsc{IsStrong} problem. We associate to an input to  \textsc{Half independent set} the game $\Delta_3(G,n/2)=(G_3,f_3,n+m+5,N_3)$ where $N_3=V\cup \{z\}$  and $(G_3,f_3)$ is the influence graph described in Figure~\ref{fig:Delta3}. Which is a variation of $\Delta_1(G,k)$ using  ideas similar to those in the reduction from the \textsc{Set Packing} problem in Theorem~\ref{the:length-width-hard}.

When $G$ has an independent set  with size at least $n/2$, $G$ also has an independent set   $X$  with $|X|= n/2$. It is easy to see that both the team $X\cup \{z\}$  and its complement are unsuccessful in $\Delta_3(G,n/2)$.   Therefore, $\Delta_3(G,n/2)$ is not strong.
Assume now  that all the independent sets in $G$ have less than $n/2$  vertices. 
Observe that, for a team $X$ in $\Delta_3(G,n/2)$ with $|X\cap V| <n/2$, its complement has at least $n/2+1$ elements in $V$ and thus it is successful. When   $|X\cap V| >n/2$ the team is  successful. Therefore we have to consider only those teams  with   $|X\cap V| =n/2$.  In such a case, we know that neither  $X\cap V$  nor $V\setminus(X\cap V)$ are independent sets. Then, by construction,  one of the sets $X$ or $N\setminus X$ must contain $z$ and is  successful while its complement  is unsuccessful. In consequence $\Delta_3(G,n/2)$ is strong.
\end{proof}

The complexity of the remaining problems is summarized in the following theorem.

\begin{theorem}\label{the:IG-par-prop}\label{the3}
For influence games, the problems {\sc IsPasser}, {\sc IsVetoer}, {\sc IsDictator}, {\sc IsCritical}, {\sc IsBlocking} and {\sc IsSwing} belong to {\sc P}.
\end{theorem}

\begin{proof}
We provide characterizations of the properties in terms of the sizes of $F(X)$ for adequate sets $X$.
Given an influence game $\Gamma=(G,w,f,q,N)$, $i\in N$ and $X\seb{}N$, we have.
\begin{itemize}
\item Player $i$ is a passer if and only if $|F(\{i\})| \geq q$.
\item Player $i$ is a vetoer if and only if $|F(N\setminus \{i\})| < q$.
\item Player $i$ is a dictator if and only if $|F(N\setminus \{i\})| < q$ and  $|F(\{i\})| \geq q$.
\item Player $i$  is critical for team $X$ if and only if $|F(X)|\geq q$ and $|F(X\setminus \{i\})|<q$.
\item Team $X$ is blocking if and only if $|F(N\setminus{}X)|<q$.
\item Team $X$ is a swing  if and only if $|F(X)|\geq q$ and there is $i\in X$ for which $|F(X\setminus\{i\})<q$.
\end{itemize}
Therefore, from Lemma~\ref{lem:Fpoly}, we get the claimed result.
\end{proof}

Now we consider the complexity of the problems related to game isomorphism and equivalence. We state here the definitions for influence games.

\begin{definition}
Let $\Gamma=(G,w,f,q,N)$ and $\Gamma'=(G',w',f',q',N')$ be two influence games with the same number of players.
$\Gamma$ and $\Gamma'$ are {\em isomorphic} if and only if there exists a {\em bijective function}, $\varphi:N\to{}N'$, such that
$|F(X)|\ge q \mbox{ if and only if }  |F(\varphi(X))|\ge q'$.
Moreover, when $N=N'$ and $\varphi$ is the {\em identity} function, then  we say that the two influence games are {\em equivalent}.
\end{definition}

The associated problems have as input two influence games and are stated as follows.
\begin{tabbing}
 xxxxxxxxxxxxxxxx    \= xxxxxxxxxxxxxxxxxxxxxxxxxxxxxxxxxxxxxxxxxxxxxxxx \kill
 {\sc Iso}:          \> Determine whether the two influence games are isomorphic.
\end{tabbing}

\begin{tabbing}
 xxxxxxxxxxxxxxxx    \= xxxxxxxxxxxxxxxxxxxxxxxxxxxxxxxxxxxxxxxxxxxxxxxx \kill
 {\sc Equiv}:        \> Determine whether the  two influence games are equivalent.
\end{tabbing}

\begin{theorem}\label{the:IG-iso-equiv}
For unweighted influence games with polynomial number of vertices, the problem {\sc Equiv} is {\sc coNP}-complete and the problem {\sc Iso} is {\sc coNP}-hard and belongs to $\Sigma^p_2$.
\end{theorem}
\begin{proof}
Membership to the corresponding complexity classes follows directly from the definition of the problems. For the hardness part we provide a reduction from the complement of the {\sc Vertex Cover} problem.
Let $G$ be a graph and consider the  influence game $\Gamma_1=\Delta_1(G,k)$ as defined before (see  Figure~\ref{fig:Delta}). Recall that the set of players is $N_1 = \{v_1,\dots,v_n,z\}$. To define the second  influence game $\Gamma_2$ we consider the weighted game  with  set of players $N_1$ and quota $q=k+1$.  The weights of the players are the following:  $w(v_i)=1$, for any $1\leq i\leq n$, and $w(z)=0$. A representation of $\Gamma_2$ as an unweighted influence game can be obtained  in polynomial time using the construction of Theorem~\ref{the:WG-UIG}.  Our reduction associates to an input to vertex cover $(G,k)$ the pair of influence games $(\Gamma_1, \Gamma_2)$.  Observe that $\Gamma_1$ is equivalent (isomorphic) to $\Gamma_2$ if and only if $G$ does not have a vertex cover of size $k$ or less.
\end{proof}

We have been unable to provide a complete classification for the {\sc Iso} problem. It remains open to show whether the problem is $\Sigma^p_2$-hard or not.  

\section{Unweighted Influence Games on Undirected Graphs}\label{sec:unIG}

In this section we analyze the complexity of the proposed problems on some particular subfamilies of unweighted influence games defined on undirected graphs.

\subsection{Maximum Influence}\label{sec:MaxI}

Here we analyze first the case with maximum influence and maximum spread, that is games of the form $\Gamma=(G,f,|V|,V)$ where $f(v)=d_G(v)$, or, as we said before, the game $\Gamma= \Gamma(G)$, for some graph $G$. 
When the graph $G$ is disconnected with connected components $C_1,\dots, C_k$,   the  associated game $\Gamma(G)$  can be analyzed from the $\Gamma(C_1), \dots, \Gamma(C_k)$. Observe that, due to maximum spread, a successful team 
must influence all the vertices in the graph. Therefore, the members of a successful team in a connected component must influence all the vertices in their component. So, a  team $X$ is successful in $\Gamma(G)$ if and only if, for any  $1\leq i\leq k$, the team $X\cap V(C_i)$ is successful in $\Gamma(C_i)$. We analyze first the case in which $G$ is connected. 
\begin{theorem}\label{the7}
In an unweighted influence game $\Gamma$ with maximum influence and maximum spread on a connected graph $G$ the following properties hold.
\begin{itemize}
\item $\Gamma$ is proper if and only if $G$ is not bipartite.
\item $\Gamma$ is strong if and only if $G$ is either a star or a triangle.
\item $\Gamma$ is decisive if and only if $G$ is a triangle.
\end{itemize}
\end{theorem}

\begin{proof}
From Lemma~\ref{lem:gammaG} we know that the successful teams of $\Gamma=\Gamma(G)$ coincide with the vertex covers of $G$. We also recall that the complement of a vertex cover is an independent set.

If $G=(V,E)$ is bipartite, let $(V_1,V_2)$ be a partition of $V$ so that $V_1$ and $V_2$ are independent sets. In such a case, both  $V_1$ and $V_2=N\setminus{}V_1$ are successful teams in $\Gamma$. Therefore, $\Gamma$ is not proper. For the opposite direction, 
 if $\Gamma$ is not proper, then the game admits 
two disjoint successful team, i.e, two disjoint vertex covers of $G$, and hence each of them must be an independent set. Thus the graph  $G$ is bipartite.

Now we prove that $\Gamma$ is not strong if and only if $G$ has at least two non-incident edges. Observe that  a graph where all edges are incident is either a triangle or a star.
If $G$ has at least two non-incident edges $e_1=(u_1,v_1)$ and $e_2=(u_2,v_2)$,  $\{u_1,v_1\}$ and $N\setminus\{u_1,v_1\}$ are both unsuccessful teams,
therefore $\Gamma$ is not strong.
When the game is not strong, there is a team $X$ such that both $X$ and $N\setminus X$ are unsuccessful.  For this to happen it must be that there is an edge uncovered by $X$ and another edge uncovered by $N\setminus X$. Thus $G$ must have two non-incident edges.

Finally, it is well known that non-bipartite graphs has at least one odd cycle, so the only non-bipartite graph with all pair of edges incidents (proper and strong) is a triangle (decisive).
\end{proof}

When the graph is disconnected, a successful team $X_i$ in the game $\Gamma(C_i)$ can be completed to a winning coalition in $\Gamma$.  Observe that, if $\overline{X}_i =V(C_i)\setminus X_i$, the set $V\setminus \overline{X}_i $ is successful in $\Gamma$ and contains $X_i$.  For an unsuccessful team $X_i$ in $\Gamma(C_i)$  both $X_i$ and  $V\setminus \overline{X}_i $ are unsuccessful in $\Gamma$. Therefore, the previous result can be extended to disconnected graphs by requesting the conditions to hold  in all the connected components of the given graph. 

\begin{corollary}
In an unweighted influence game $\Gamma$ with maximum influence and maximum spread on a graph $G$ the following properties hold.
\begin{itemize}
\item $\Gamma$ is proper if and only if  all the connected components of $G$ are not bipartite.
\item $\Gamma$ is strong if and only if all the connected components of $G$ are either a star or a triangle.
\item $\Gamma$ is decisive if and only if all the connected components of $G$ are triangles.
\end{itemize}
Furthermore, the problems \textsc{IsProper}, \textsc{IsStrong} and \textsc{IsDecisive} belong to P for unweighted influence games with maximum influence and maximum spread.
\end{corollary}

In regard to the complexity of the two main game measures we have the following result.

\begin{theorem}\label{the8}
For unweighted influence games with maximum influence and maximum spread on a connected graph $G$, computing 
{\sc Length} is {\sc NP}-hard. Computing {\sc Width} of $\Gamma(G)$ can be done in polynomial time even when $G$ is disconnected. 
\end{theorem}

\begin{proof}
As before we use the fact that $\Gamma(G)$ can be computed in polynomial time. Furthermore, from Lemma~\ref{lem:gammaG},  $\textsc{Length}(\Gamma(G))$ is  the minimum size of a vertex cover of $G$. Therefore {\sc Length} is  {\sc NP}-hard.

We prove that \textsc{Width} can be computed in polynomial time by a case analysis. 
If $G$ is just an isolated vertex or just one edge, the empty set is the unique unsuccessful team, thus  $\textsc{Width}(\Gamma)=0$. 
Otherwise, either $G$ has no edges or has at least one edge and an additional  vertex. 
In the first case, the graph is an independent set with at least two vertices. Assume that $u\in V$, then $V\setminus \{u\}$ is unsuccessful and we conclude that  $\textsc{Width}(\Gamma)=n-1$.  

In the second case $G$ has at least one edge $e=(u,v)$ and  $V\setminus\{u,v\}$ is non-empty. We have again two cases, either $G$  has an isolated vertex $u$ or all the connected components of $G$ have at least one edge. 
When $u$ is an isolated vertex the team $V\setminus\{u\}$ is unsuccessful, therefore  $\textsc{Width}(\Gamma)=n-1$.
When  all the connected components of $G$ have at least one edge, any team with $n-1$ nodes is a vertex cover, thus $\textsc{Width}(\Gamma)<n-1$.  Observe that the set $V\setminus \{u,v\}$ is not empty and, furthermore it does not cover the edge $e$, thus we have an unssuccesful team with $n-2$ vertices. Thus, in this case $\textsc{Width}(\Gamma)=n-2$.  

As the classification can be
checked trivially in polynomial time we get the claimed result. 
\end{proof}

For the case of maximum influence but not maximum spread, that is influence games of the form $(G,f,q,V)$ where $f(v)=d_G(v)$ and $q< n$, the game cannot be directly analyzed from the games on the connected components, as the total quota can be fulfilled in different ways by the agents influenced in each component. Nevertheless, as the influence is maximum, any set of vertices $X$ can influence another vertex $u$ only when all the neighbors of $u$ are included in $X$, alternatively when $u$ becomes an isolated vertex after removing $X$. This leads to the following characterization of  the successful teams.

\begin{lemma}\label{lem1}
In an unweighted influence game with maximum influence $\Gamma=(G,d_G,q,V)$ where $G$ has no isolated vertices,
$X\seb V$ is a successful team if and only if removing $X$ from $G$ leaves at least $q-|X|$ isolated vertices.
\end{lemma}

This characterization gives rise to the following problem:

\begin{tabbing}
 xxxxxxxxxxxxxxxx    \= xxxxxxxxxxxxxxxxxxxxxxxxxxxxxxxxxxxxxxxxxxxxxxxx \kill
 {\sc AreIsolated}:  \> Given a graph $G=(V,E)$ and $q,k\in\matN$.\\
		     \> Is there $S\seb V$ such that $|S|\leq k$ and removing $S$ from $G$\\
		     \> there are at least $q-k$ isolated vertices?
\end{tabbing}

Observe that for $q=n$ we have that the solution $S$ in the previous problem must be a vertex cover, and thus the {\sc AreIsolated} problem  is {\sc NP}-hard.

\begin{theorem}\label{the9}
For influence games $\Gamma$ with maximum influence, {\sc Length} is {\sc NP}-hard and {\sc Width} belongs to {\sc P}.
\end{theorem}

\begin{proof}
The hardness result follows from the previous observation. Observe that computing the minimum size of a solution to the  {\sc AreIsolated} problem is equivalent to compute the {\sc Length} of the game   $\Gamma=(G,d_G,q,V)$ and thus the later problem is {\sc NP}-hard.

When  computing the {\sc Width} of $\Gamma=(G,d_G,q,V)$  we want to maximize the size of the unsuccessful teams. Therefore, we can restrict ourselves to analyze only  unsuccessful teams $X$ for which $F(X)=X$. 
We have that $X$ is an unsuccessful team with $F(X)=X$ if and only if  $|X|< q$ and every non isolated vertex in $V\setminus X$ remains non isolated in $G[V\setminus X]$ (the subgraph induced by $N\setminus X$).

We consider first the case in which  $G$ has no isolated vertices.  We first solve the problem of deciding whether, for a given $\alpha$,  it is possible to discard $\alpha$ nodes from $G$ without leaving isolated vertices. For doing so we sort the sizes of the connected components of $G$  in increasing order of size. As $G$ has no isolated vertices all the connected components have at least two vertices.  Assume that $G$ has $k$ connected components $C_1,\dots, C_k$ with sizes  $2\leq w_1\leq w_2 \leq \dots \leq w_k$. 

When  $w_k=2$, all the connected components have exactly two vertices. Therefore, if $\alpha$ is even and at most $n$, we can discard the $\alpha$ vertices in the first $\alpha/2$ components, without leaving isolated vertices. Otherwise, the removal of any set of size $\alpha$ will leave at least one isolated vertex.

When  $w_k> 2$. We compute the first value $j$ for which  $\sum_{i=1}^j w_i \leq \alpha$ but $\sum_{i=1}^{j+1} w_i > \alpha$. Let $\beta=\sum_{i=1}^j w_i$. Let $S_j$ be the set of vertices in the first $j$-components.
If $\beta = \alpha$, $S_j$  can be removed without leaving isolated vertices. 
When $\beta< \alpha$  we have two cases: 
\begin{itemize}
\item[(1)]  $w_{j+1}> \alpha -\beta+1$.  Let $C\subset C_{j+1}$ be a set with  $w_{j+1}-(\alpha-\beta)$ vertices such that $G[C]$ is connected. The vertices in   $S_j$  together with the $\alpha-\beta$ vertices of $C_{j+1}$ not in $C$ can be removed without leaving any isolated vertex. 
\item[(2)] $w_{j+1} \leq \alpha -\beta+1$. By construction, $\alpha < \beta + w_{j+1}$, thus  $w_{j+1}=\alpha -\beta+1$. If $j+1<k$, removing the vertices in $S_j$  together with  $\alpha -\beta-1$ vertices from the $j+1$-th component (as in case (1)) and one additional vertex from the $k$-th component leaves no isolated vertices. If $j+1=k$, the removal of any set of size $\alpha$ will leave at least one isolated vertex. 
\end{itemize}

The previous characterization can be decided in polynomial time for any value of $\alpha$.  
By performing the test for $\alpha= q-1, q-2, \dots, 1$ we can compute in polynomial time the maximum value of $\alpha$ ($\alpha_{max}$) for which $\alpha$ nodes can be discarded without leaving isolated vertices. As the {\sc Width} of the game is just $\alpha_{max}$ we get the desired result for graphs without isolated vertices.

When $G$ has $n_0$ isolated vertices, we consider the graph $G'$ obtained from $G$ by removing all the isolated vertices.  Note that for a team $X$ with $X=F(X)$ and any set $Y$ of isolated vertices we have that $F(X\cup Y)=X\cup Y$, thus $\textsc{Width} (\Gamma) = \min\{\textsc{Width} (\Gamma')+n_0, q-1\}$. Therefore \textsc{Width} can be computed in polynomial time. 

\end{proof}

\subsection{Minimum Influence}\label{sec:MinI}

Let be $\Gamma=(G,1_V,q,N)$ where $1_V(v)=1$ for any $v\in{}V$.
Observe that if $G$ is connected, the game has a trivial structure as any non-empty vertex subset of $N$ is a successful team. For the disconnected case we can analyze the game considering an instance of the \textsc{Knapsack} problem. Assume that $G$ has $k$ connected components, $C_1,\dots,C_k$. Without loss of generality, we assume that all the connected components of $G$ have  non-empty intersection with $N$. For $1\le{}i\le{}k$, let $w_i=|V(C_i)|$ and $n_i=|V(C_i)\cap N|$. 

\begin{lemma}\label{lem2}
If a successful team $X$ is minimal then it has  at most one  node in each connected component. Minimal successful teams are in a many-to-one correspondence with the minimal winning coalitions of the weighted game $[q;w_1,\dots,w_k]$.
\end{lemma}

Moreover, we have the following result.

\begin{theorem}\label{the3:Min-Inf}
For unweighted influence games with minimum influence, the problems
{\sc Length}, {\sc Width}, {\sc IsProper}, {\sc IsStrong} and {\sc IsDecisive} belong to {\sc P}.
\end{theorem}

\begin{proof}
Let $\Gamma=(G,1_V,q,N)$ be an  unweighted influence game with minimum influence.

First we prove that {\sc Length} can be computed in polynomial time. 
Assume that the connected components of $G$ are sorted in such a way that 
 $w_1\geq \dots\geq w_k$. To minimize the size of a winning coalition we consider only those coalitions with at most one player in a connected component.
Observe that, the $\textsc{Length}(\Gamma)$ is the minimum $j$ for which $\sum_{i=1}^j w_i \geq q$ but $\sum_{i=1}^{j-1} w_i < q$.  Of course this value can be computed in polynomial time.

For computing \textsc{Width} observe that an unsuccessful team of maximum size  can be obtained by computing a selection  $S\subseteq \{1, \dots,k\}$ of connected components in such a way that $\sum_{i\in S} w_i < q$ and  $\sum_{i\in S} n_i$ is maximized.  Computing such selection is equivalent to solving a \textsc{Knapsack} problem on a set of $k$ items, item $i$ having weight $w_i$ and value $n_i$, and setting the knapsack capacity to $q$.   As the \textsc{Knapsack} problem can be solved in pseudo polynomial time and, in our case, all the weights and values are at most $n$, we conclude that  {\sc Width} can be computed in polynomial time.

Now we prove that the {\sc IsStrong} problem belongs to {\sc P}. Observe that in order to minimize the influence of the complementary of a team $X$ it is enough to consider only those  teams $X$ that contain all or none of the players in a connected component. 
Let $w_N=\sum_{i=1}^k w_i$, and let $\alpha_{max}$ be the maximum $\alpha\in \{0,\dots,q-1\}$ for which there is a set $S\seb\{1,\ldots,k\}$ with $\sum_{i\in S} w_i = \alpha$.
Note that $\alpha$ can be zero and thus $S$ can be the empty set. Observe that $\Gamma$  is strong if and only if $w_N-\alpha_{max} \geq q$.
The value $\alpha_{max}$ can be computed by solving several instances of the \textsc{Knapsack} problem.  As the weights are at most $n$, the value can be obtained in polynomial time.

Now we prove that the {\sc IsProper} problem can be solved in polynomial time.
To check whether the game is not proper it is enough to show that  there is a winning coalition whose complement is also winning. 
 For doing so we separate the connected components in two sets: those containing one player and those containing more that one player. Let $A=\{i\mid n_i=1\}$ and $B=  \{i\mid n_i>1\}$. Let $N_A=\cup_{i\in A} (N \cap V(C_i))$ and $N_B=N\setminus N_A$.  Let $w_A=\sum_{i\in A} w_i$ and $w_B=w_N-w_A$.
As all the components in $B$ have at least two vertices, we can find a set $X\subseteq N_B$ such that $|F(X)|=|F(N_B\setminus X)|=w_B$. Thus if $w_B\geq q$ the game is not proper. When $w_B<q$ the game is proper if and only if the influence game $\Gamma'$  played on the graph formed by the connected components belonging to  $A$ and quota $q'=q-w_B$ is proper.  Observe that $\Gamma'$ is  equivalent to the weighted game with a player for each component in $i\in A$ with associated weight $w_i$ and quota $q'$. 

Let $\alpha_{min}$ be the minimum $\alpha\in\{q',\dots,w_A\}$ for which there is a set $S\seb A$ with $\sum_{i\in S} w_i = \alpha$.
  Observe that $\Gamma'$ is proper if and only if $w_A-\alpha_{min} < q'$. 
The value $\alpha_{min}$ can be computed by solving several instances of the \textsc{Knapsack} problem having item  weights polynomial in $n$.  Therefore, $\alpha_{min}$ can be computed in  polynomial time and the claim follows.
\end{proof}

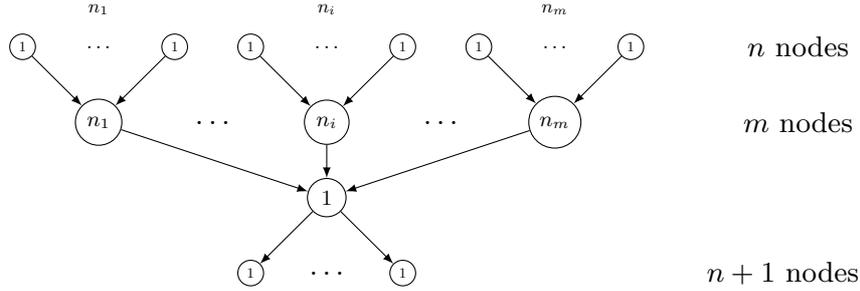
\begin{figure}[t]
\centering
\begin{tikzpicture}[every node/.style={circle,scale=0.6}, >=latex]
\node[draw](d)   at (0,3)  {1};
\node[scale=1.3] at (1,3.5){{$n_1$}};
\node[scale=1.4] at (1,3){{$\ldots$}};
\node[draw](f)   at (2,3)  {1};
\node[draw](g)   at (3,3)  {1};
\node[scale=1.3] at (4,3.5){{$n_i$}};
\node[scale=1.4] at (4,3){{$\ldots$}};
\node[draw](i)   at (5,3)  {1};
\node[draw](j)   at (6,3)  {1};
\node[scale=1.3] at (7,3.5){{$n_m$}};
\node[scale=1.4] at (7,3){{$\ldots$}};
\node[draw](l)   at (8,3)  {1};

\node[draw,scale=1.5](a) at (1,2)  {$n_1$};
\node[scale=2]  	 at (2.5,2){$\ldots$};
\node[draw,scale=1.5](b) at (4,2)  {$n_i$};
\node[scale=2]   	 at (5.5,2){$\ldots$};
\node[draw,scale=1.5](c) at (7,2)  {$n_m$};

\node[draw,scale=1.5](m) at (4,1) {$1$};

\node[draw](q)   at (3,0)  {1};
\node[draw](s)   at (5,0)  {1};
\node[scale=2]   at (4,0)  {$\ldots$};

\draw[<-] (a) to node {}(d);
\draw[<-] (a) to node {}(f);
\draw[<-] (b) to node {}(g);
\draw[<-] (b) to node {}(i);
\draw[<-] (c) to node {}(j);
\draw[<-] (c) to node {}(l);
\draw[->] (a) to node {}(m);
\draw[->] (b) to node {}(m);
\draw[->] (c) to node {}(m);
\draw[->] (m) to node {}(q);
\draw[->] (m) to node {}(s);

\node[scale=2] at (10.2,3){\small{$n$ nodes}};
\node[scale=2] at (10.2,2){\small{$m$ nodes}};
\node[scale=2] at (10,0){\small{$n+1$ nodes}};
\end{tikzpicture}
\vspace{-5ex}\caption{The simple game whose set of players $N=\{1,\ldots,n\}$ admits a partition $N_1,\ldots,N_m$ in such a way that $\cW=\{S\subseteq{}N\,;\,\exists{}N_i\mbox{\ with\ }N_i\subseteq{}S\}$ has exponential dimension, $n_1\cdot\ldots\cdot{}n_{m-1}$~\cite{FP01}, but this game admits a polynomial unweighted influence graph $(G,f)$ with respect to $n$ for the corresponding unweighted influence game $(G,f,n+1,N)$.\label{fig:SGExpDim}}
\end{figure}

\section{Conclusions}

In this paper we have considered an influence spread model in social networks as a mechanism for emergence of cooperation. Our model is based on the linear threshold model for influence spread. In such a context we use parameters and properties derived from voting systems (simple games) to analyze the properties of the system.
This point of view has given rise to the introduction of a new family of simple games, called {\em influence games}. We have shown that influence games are enough expressive to represent the complete family of simple games. Interestingly enough we have shown that representations, as influence games, of the union or intersection of two influence games can be computed in polynomial time. We have also shown that influence games in which the number of agents is polynomial in the number of players is a proper subset of simple games.

The remaining contribution of this paper  concerns the computational complexity of problems related to several parameters and properties. Our results are summarized in Table~\ref{tab:results}.
As a side result we have been able to show that computing the width of a game given by the set of winning or minimal winning coalitions can be computed in polynomial time, which was posted as an open problem in~\cite{Azi09}.
Finally, we have analyzed two extreme cases for the required level of influence, maximum and minimum, showing that the computational complexity of some of the considered problems changes, becoming in  general more tractable. In particular we have shown a graph characterization of the proper and strong properties in terms of graph properties for influence graphs with maximum influence and maximum spread.

We have studied the isomorphism and equivalence problems for influence games. It remains open to show whether the {\sc Iso} problem is or not $\Sigma_2^p$-complete. The equivalence problem has been studied for weighted and multiple weighted game~\cite{EGGW08} (see Table~\ref{tab:results}). To the best of our knowledge, the complexity of the \textsc{Iso} problem remains open  for weighted games as well as  for any of the families of simple games defined by  boolean combinations of  weighted games considered in this paper.

There are many open lines for future research, here we mention a few. On the general topic of spreading influence, it will be interesting to analyze the properties of the influence games defined through other influence spread mechanisms, in particular for randomized models like the independent cascade model ~\cite{KKT03}.
Besides the extreme situations previously mentioned, there are many other natural rules to study, such as the majority rule, when individuals are convinced when a majority of their neighbors are. Also, analyzing the properties of particular types of graphs arising in social networks or other type of organizations in which influence spread can take part in the decision process.

Simple games are defined by monotonic families that can be defined succinctly by monotonic functions. Therefore, there are several  circuit or graph based formalisms that can be used to  represent simple games. Among them the \emph{binary decision diagrams} (BDDs)  has been used as an alternative representation of  simple games~\cite{Bol11}. This succinct form of representation is most commonly used to represent Boolean functions~\cite{Ake78}. BDDs have been used to study several properties of simple games, regular games and weighted games~\cite{BB12}. However, except for some results regarding the computation of power indices, for subfamilies of simple games~\cite{Bol11}, the complexity of most of the problems studied in this paper remain as an open problem.

Lastly, even though we have shown that all games with polynomial dimension or polynomial codimension can be represented as weighted influence games in polynomial time (i.e., they admit weighted influence graphs with polynomial number of agents),
a fundamental open question is determining which simple games can be represented as an (unweighted) influence game with polynomial number of agents. In particular it remains open to know whether there are games with exponential dimension that also require an exponential number of players in any representation as influence games. 
In this line, we know that the simple game with exponential dimension with respect to the players of Section 2 in~\cite{FP01} can be represented by a unweighted influence game in polynomial time with respect to the number of players (see Figure~\ref{fig:SGExpDim}). 
Another candidate is the simple game with exponential dimension of Theorem 8 in~\cite{EGGW08} for which  we have been unable to show whether it can be represented by a (unweighted) influence game with polynomial number of agents or not.

\section*{Acknowledgments} We thank S.~Kurz for pointing out the proof of Theorem~\ref{teo:polyIG}.

\bibliographystyle{plain}
\bibliography{Bib}

\end{document}